\documentclass[10pt,journal]{IEEEtran}
\usepackage[utf8]{inputenc}
\usepackage{amsmath}
\usepackage{amssymb}
\usepackage{graphicx} 
\usepackage{cuted, nccmath}
\usepackage{epsfig}
\usepackage{epstopdf}
\usepackage{mathbbol}
\usepackage{subfigure}
\usepackage{epsf}
\usepackage{bm}
\usepackage{amsthm, amsfonts, tikz, color, environ} 
\usepackage{stfloats}
\usepackage{float}
\usepackage{stmaryrd}
\usepackage[noadjust]{cite}
\usepackage{footnote} 
\usepackage{enumitem}
\newtheorem{theorem}{\textit{Theorem}}
\newtheorem{corollary}{\textit{Corollary}}
 
\usepackage{algorithm}
\usepackage{algorithmic}
\DeclareMathOperator\erf{erf}
\usepackage{array,multirow,makecell}
\setcellgapes{1pt}
\makegapedcells
\newcolumntype{R}[1]{>{\raggedleft\arraybackslash }b{#1}}
\newcolumntype{L}[1]{>{\raggedright\arraybackslash }b{#1}}
\newcolumntype{C}[1]{>{\centering\arraybackslash }b{#1}} 
\begin{document} 
\title{Measurements-Based Channel Models for Indoor LiFi Systems}
\author{Mohamed Amine Arfaoui$^{*}$,
        Mohammad Dehghani Soltani,
        Iman Tavakkolnia,
        Ali Ghrayeb, \\
        Chadi Assi, 
        Majid Safari, 
        and Harald Haas \vspace{-0.7cm}
\thanks{M. A. Arfaoui and C. Assi are with Concordia Institute for Information Systems Engineering (CIISE), Concordia University, Montreal, Canada, e-mail:\{m\_arfaou@encs, assi@ciise\}.concordia.ca.}
\thanks{M. D. Soltani, I. Tavakkolnia, M. Safari, and H. Haas are with the LiFi Research and Development Centre, Institute for Digital Communications, School of Engineering, University of Edinburgh, UK. e-mail: \{m.dehghani, i.tavakkolnia, majid.safari, h.haas\}@ed.ac.uk.}
\thanks{A. Ghrayeb is with the Electrical and Computer Engineering (ECE) department, Texas A$\&$M University at Qatar, Doha, Qatar, e-mail: ali.ghrayeb@qatar.tamu.edu.}
\thanks{$^*$\textit{Corresponding author: M.A. Arfaoui, m$\_$arfaou@encs.concordia.ca}}
}
\maketitle 
\thispagestyle{plain}
 \begin{abstract} Light-fidelity (LiFi) is a fully-networked bidirectional optical wireless communication (OWC) that is considered a promising solution for high-speed indoor connectivity. Unlike in conventional radio frequency wireless systems, the OWC channel is not isotropic, meaning that the device orientation affects the channel gain significantly. However, due to the lack of proper channel models for LiFi systems, many studies have assumed that the receiver is vertically upward and randomly located within the coverage area, which is not a realistic assumption from a practical point of view. In this paper, novel realistic and measurement-based channel models for indoor LiFi systems are proposed. Precisely, the statistics of the channel gain are derived for the case of randomly oriented stationary and mobile LiFi receivers. For stationary users, two channel models are proposed, namely, the modified truncated Laplace (MTL) model and the modified Beta (MB) model. For LiFi users, two channel models are proposed, namely, the sum of modified truncated Gaussian (SMTG) model and the sum of modified Beta (SMB) model. Based on the derived models, the impact of random orientation and spatial distribution of LiFi users is investigated, where we show that the aforementioned factors can strongly affect the channel gain and system performance.
\end{abstract} 
\IEEEpeerreviewmaketitle 
\section{Introduction}
\subsection{Motivation}
\indent The total data traffic is expected to become about 49 exabytes per month by 2021, while in 2016, it was approximately 7.24 exabytes per month \cite{Intro1}. With this drastic increase, the fifth generation (5G) networks and beyond must urgently provide high data rates, seamless connectivity, robust security and ultra-low latency communications \cite{Intro2,Intro3,Intro4}. In addition, with the emergence of the internet-of-things (IoT) networks, the number of connected devices to the internet is increasing dramatically \cite{Intro5,Intro6}. This fact implies not only a significant increase in data traffic, but also the emergence of some IoT services with crucial requirements. Such requirements include high data rates, high connection density, ultra reliable low latency communication (URLLC) and security. However, traditional radio-frequency (RF) networks, which are already crowded, are unable to satisfy these high demands \cite{Intro7}. Network densification \cite{Intro8,Intro9} has been proposed as a solution to increase the capacity and coverage of 5G networks. However, with the continuous dramatic growth in data traffic, researchers from both industry and academia are trying to explore new network architectures, new transmission techniques and new spectra to meet these demands. 
\\
\indent Light-fidelity (LiFi) is a novel bidirectional, high speed and fully networked wireless communication technology, that uses visible light as the propagation medium in the downlink for the purposes of illumination and communication. It can use infrared in the uplink so that the illumination constraint of a room remains unaffected, and also to avoid interference with the visible light in the downlink \cite{haas2015lifi}. LiFi offers a number of important benefits that have made it favorable for future technologies. These include the very large, unregulated bandwidth available in the visible light spectrum (more than 2600 times greater than the whole RF spectrum), high energy efficiency \cite{tavakkolnia2018energy}, the straightforward deployment that uses off-the-shelf light emitting diode (LED) and photodiode (PD) devices at the transmitter and receiver ends, respectively, and enhanced security as light does not penetrate through opaque objects \cite{VLC8}. However, one of the key shortcomings of the current research literature on LiFi is the lack of appropriate statistical channel models for system design and handover management purposes. 
\subsection{Literature Review}
\indent Some statistical channel models for stationary and uniformly distributed users were proposed in \cite{yin2016performance,gupta2017cascaded,yapici2018non}, where a fixed incidence angle was assumed in \cite{yin2016performance,gupta2017cascaded} and a random incidence angle was assumed in \cite{yapici2018non}. However, accounting for mobility, which is an inherent feature of wireless networks, requires a more realistic and non-uniform model for users' spatial distribution. Several mobility models, such as the random waypoint (RWP) model, have been proposed in the literature to characterize the spatial distribution of mobile users for indoor RF systems \cite{govindan2011probability,aalo2016effect}. However, these studies were limited to RF spectrum where statistical fading channel models were used. Recently, \cite{gupta2018statistics,arfaoui2019SNR} employed the RWP mobility model to characterize the signal-to-noise ratio (SNR) for indoor LiFi systems. In \cite{gupta2018statistics}, the device orientation was assumed constant over time, which is not a realistic scenario, whereas in \cite{arfaoui2019SNR}, the incidence angle of optical signals was assumed to be uniformly distributed, which is not a proper model for the incidence angle, since it does not account for the actual statistics of device orientation.  \\
\indent Device orientation can significantly affect the users’ throughput. The majority of studies on OWC assume that the device always faces vertically upward. This assumption may have been driven by the lack of having a proper model for orientation, and/or to make the analysis tractable. Such an assumption is only accurate for a limited number of devices (e.g., laptops with a LiFi dongle), while the majority of users use devices such as smartphones, and in real-life scenarios, users tend to hold their device in a way that feels most comfortable. Such orientation can affect the users’ throughput remarkably and it should be analyzed carefully. Even though a number of studies have considered the impact of random orientation in their analysis \cite{APselectionLiFi,MDSHandover,ArdimasLiFiRF,wang2017impact,wang2017improvement,wang2011performance,matrawy2016optimum,eroglu2017impact}, all these studies assume a predefined model for the random orientation of the receiver. However, little or no evidence is presented to justify the assumed models. Nevertheless, none of these studies have considered the actual statistics of device orientation and have mainly assumed uniform or Gaussian distribution with hypothetical moments for device orientation. Recently, and for the first time, experimental measurements were carried out to model the polar and azimuth angles of the user's device in \cite{soltani2018modeling,purwita2018WCNC,Zhihong_VTCfall_2018,MDS2019Thesis}. It is shown that the polar angle can be modeled by either a truncated Laplace distribution for the case of stationary users or a truncated Gaussian distribution for the case of mobile users, while the azimuth angle follows a uniform distribution for both cases. Motivated by these results, the impact of the random receiver orientation on the SNR and the bit error rate (BER) was studied for indoor stationary LiFi users in \cite{soltani2018impact}. \\ 
\indent Solutions to alleviate the impact of device random orientation on the received SNR and throughput were proposed in \cite{mohammad2018optical,tavakkolnia2019mimo,ChengICCW19}. In \cite{mohammad2018optical}, the impact of the random receiver orientation, user mobility and blockage on the SNR and the BER was studied for indoor mobile LiFi users. Then, simulations of BER performance for spatial modulation using a multi-directional receiver configuration with consideration of random device orientation was evaluated. In \cite{tavakkolnia2019mimo}, other multiple-input multiple-output (MIMO) techniques in the presence of random orientation were studied. The authors in \cite{ChengICCW19}, proposed an omni-directional receiver which is not affected by the device random orientation. It is shown that the omni-directional receiver reduces the SNR fluctuations and improves the user throughput remarkably. All these studies emphasize the significance of incorporating the random spatial distribution of LiFi users along with the random orientation of LiFi devices into the analysis. However, proper statistical channel models for indoor LiFi systems that encompass both the random spatial distribution and the random device orientation of LiFi users were not derived in the literature, which is the focus of this work.
\subsection{Contributions and Outcomes}
\indent Against the above background, we investigate in this paper the channel statistics of indoor LiFi systems. Novel realistic and measurement-based channel models for indoor LiFi systems are proposed, and the proposed models encompass the random motion and the random device orientation of LiFi users. Precisely, the statistics of the line-of-sight (LOS) channel gain are derived for stationary and mobile LiFi users with random device orientation, using the measurements-based models of device orientation derived in \cite{soltani2018modeling}. For stationary LiFi users, the model of randomly located user is employed to characterize the spatial distribution of the LiFi user, and the truncated Laplace distribution is used to model the device orientation. For mobile LiFi users, the RWP mobility model is used to characterize the spatial distribution of the user and the truncated Gaussian distribution is used to model the device orientation. In light of the above discussion, we may summarize the paper contributions as follows.
\begin{itemize}
    \item For stationary LiFi users, two channel models are proposed, namely the modified truncated Laplace (MTL) model and the modified Beta (MB) model. For mobile LiFi users, also two channel models are proposed, namely the sum of modified truncated Gaussian (SMTG) model and the sum of modified Beta (SMB) model. The accuracy of the derived models is then validated using the Kolmogorov-Smirnov distance (KSD) criterion. 
    \item The BER performance of LiFi systems is investigated for both cases of stationary and mobile users using the derived statistical channel models. We show that the random orientation and the random spatial distribution of LiFi users could have strong effect on the error performance of LiFi systems. 
    \item We propose a novel design of indoor LiFi systems that can alleviate the effects of random device orientation and random spatial distribution of LiFi users. We show that the proposed design is able to guarantee good error performance for LiFi systems under the realistic behaviour of LiFi users.
    \item The proposed statistical LiFi channel models are of great significance. In fact, any LiFi transceiver design, to be efficient, it needs to incorporate the channel model into the design. Therefore, having realistic channel models will help in designing realistic LiFi transceivers.
\end{itemize}
\subsection*{Outline and Notations}
\indent The rest of the paper is organized as follows. The system model is presented in Section II. Section III presents the exact statistics of the LOS channel gain. In Sections IV, statistical channel models for stationary and mobile LiFi users are proposed. Finally, the paper is concluded in Section V and future research directions are highlighted. \\ 
\indent The notations adopted throughout the paper are summarized in Table \ref{T1}. In addition, for every random variable $X$, $f_X$ and $F_X$ denote the probability density function (PDF) and the cumulative distribution function (CDF) of $X$, respectively. The function $\delta (\cdot)$ denotes the Dirac delta function. The function $\mathcal{U}_{[a,b]}\left( \cdot \right)$ denotes the the unit step function within $[a,b]$, i.e., for all $x \in \mathbb{R}$, $\mathcal{U}_{[a,b]}\left( x \right) = 1$ if $x \in [a,b]$, and 0 otherwise.
\begin{table}[t]
\centering
\caption{Table of Notations}
\renewcommand{\arraystretch}{1.2} 
\setlength{\tabcolsep}{0.2cm} 
\begin{tabular}{| l | l |}
  \hline 
  \multicolumn{2}{|c|}{\textbf{System Geometry}} \\
  \hline 
  \hline
  $R$ & Radius of the LiFi attocell \\ 
  \hline
  $R_{\rm e}$ & Radius of the large area \\ 
  \hline
  $h_{\rm a}$ & Height of the AP \\ 
  \hline
  $h_{\rm u}$ & Height of the LiFi receiver \\ 
  \hline 
  \hline 
  \multicolumn{2}{|c|}{\textbf{LiFi Channel Parameters}} \\
  \hline 
  \hline 
  $H$ & LOS channel gain \\ 
  \hline 
  $r$ & Polar distance of the LiFi user \\ 
  \hline 
  $\alpha$ & Polar angle of the LiFi user \\ 
  \hline 
  $d$ & Distance from the AP to the LiFi user \\
  \hline 
  $\Omega$ & Angle of the direction facing the LiFi user \\ 
  \hline 
  $\theta$ & Elevation angle of the LiFi receiver \\ 
  \hline
  $\Psi$ & Angle of incidence \\ 
  \hline
  $\Psi_c$ & Field of view \\
  \hline 
\end{tabular} 
\label{T1}
\end{table}
\section{System Model} 
\begin{figure}[t]
\centering     
\includegraphics[width=0.65\linewidth]{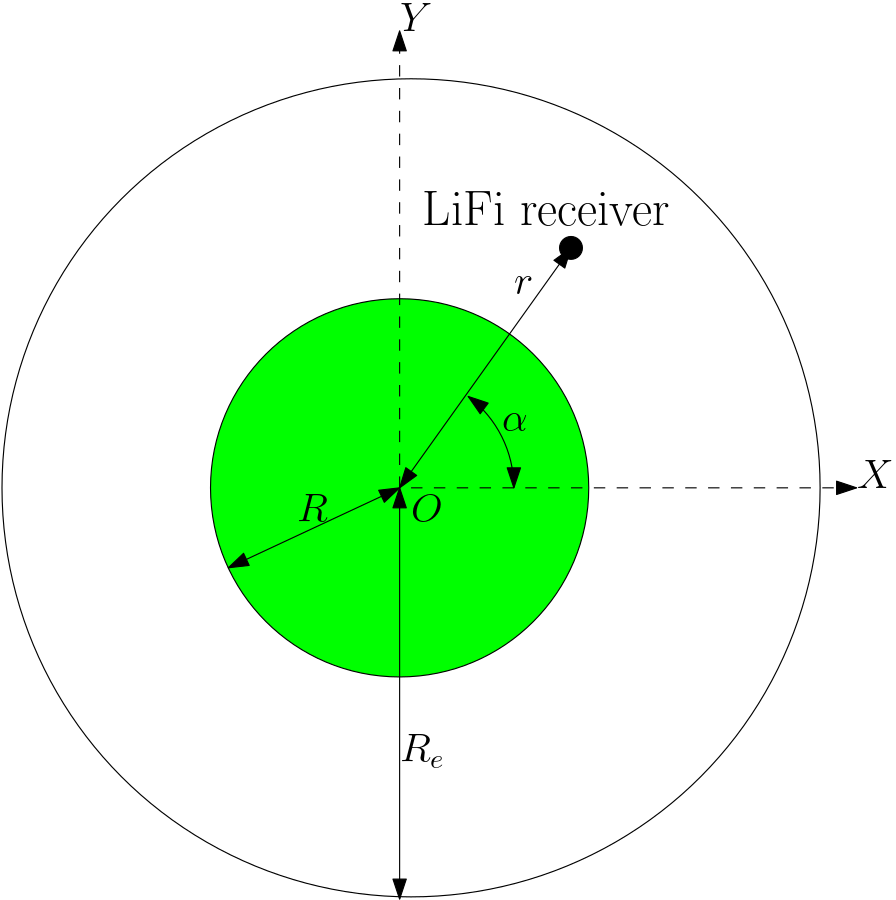}
\caption{Top view of a LiFi attocell which is concentric with a larger circular area.}
\label{fig:celldesign}
\end{figure}
Consider the indoor LiFi cellular system shown in Fig.~\ref{fig:celldesign}, which consists of a LiFi attocell with radius $R$ (green attocell), that is equipped with a single access-point (AP) installed at height $h_{\rm a}$ from the ground. The LiFi attocell is concentric with a larger circular area with a radius $R_{\rm e}$ ($R \leq R_{\rm e}$), within which a LiFi user may be located. The user equipment (UE) is equipped with a single PD that is used for communication with the AP. Assuming that the global coordinate system $\left(O,X,Y,Z \right)$ is cylindrical, the coordinates of the UE are given by $\left(r,\alpha,h_{\rm u} \right)$, where $r \in [0,R_{\rm e}]$ is the polar distance, $\alpha \in [0,2 \pi]$ is the polar angle and $h_{\rm u} \in \left[0, h_{\rm a}\right]$ is the height of the LiFi receiver. The user is assumed to hold the UE within a close distance of the body. Therefore, the polar coordinates $(r,\alpha)$ of the UE are assumed exactly the same as those of the LiFi user. However, this is not the case for the height $h_{\rm u}$, since it depends mainly on the activity of the LiFi user, i.e., either stationary (sitting activity) or mobile (walking activity). Furthermore, in this communication model, the UE can be connected to the AP if it is located inside the LiFi attocell, i.e., when $r \leq R$. In this case, the received signal at the LiFi receiver at each channel use is expressed as 
\begin{equation}
Y = HS+N, 
\end{equation}
where $H$ is the downlink channel gain, $S$ is the transmitted signal and $N$ is an additive white Gaussian noise (AWGN) that is $\mathcal{N}(0,\sigma^2)$ distributed. Since LiFi signals should be positive valued and satisfy a certain peak-power constraint \cite{arfaoui2016secrecy}, we assume that $0 \leq S \leq A$, where $A \in \mathbb{R}_+$ denotes the maximum allowed signal amplitude. \\ 
\indent The channel gain $H$ is the sum of a LOS component and a non-light-of-sight (NLOS) component resulting from reflections of walls. However, it was observed in \cite{zeng2009high} that, for indoor LiFi scenarios, the  optical  power  received  from reflected signals is negligible compared to the LOS component, especially if the LiFi receiver is far away from the walls or is located  close to the cell center. In this case, the contribution of the NLOS component is very small compared to that of the LOS component. Based on this, only the LOS component of $H$ is considered, the channel gain $H$ is expressed as \cite{zeng2009high}
\begin{equation}
\label{ChannelGain}
H = H_0  \frac{\cos(\phi)^m\cos(\psi)}{d^2} \text{rect} \left( \frac{\psi}{\Psi_c} \right), 
\end{equation}
where, as shown in Fig.~\ref{fig:systmodel1}, $m$ is the order of the Lambertian emission that is given by $m = \frac{-\log (2)}{\log(\cos(\phi_{1/2}))}$, such that $\phi_{1/2}$ represents the semi-angle of a LED; $d = \sqrt{r^2 + \left(h_{\rm a} - h_{\rm u} \right)^2}$ is the distance between the AP and the UE; $\phi \in [0, \phi_{1/2}]$ is the radiation angle; $\psi \in [0, \pi]$ is the incidence angle and $\Psi_c$ is the field of view of the PD. In \eqref{ChannelGain}, $H_0$ is
\begin{equation}
    H_0 = \rho R_p \frac{(m+1)}{2 \pi} \frac{n_c^2 A_{g}}{\sin(\Psi_c)^2},
\end{equation}
where $\rho$ is the electrical-to-optical conversion factor, $R_p$ is the PD responsivity, $A_{g}$ is the geometric area of the PD and $n_c$ is the refractive index of the PD's optical concentrator. \\ 
\begin{figure}[t]
\centering     
\includegraphics[width=0.75\linewidth]{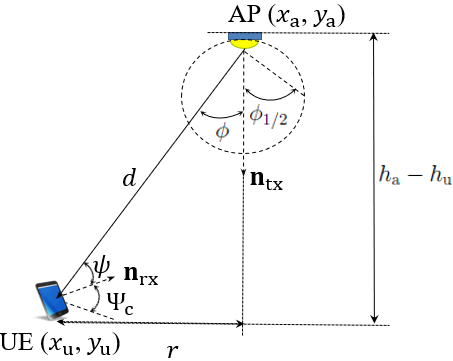}
\caption{Description of the indoor LiFi communication link.}
\label{fig:systmodel1}
\end{figure}
\begin{figure}[t]
\centering     
\includegraphics[width=0.75\linewidth]{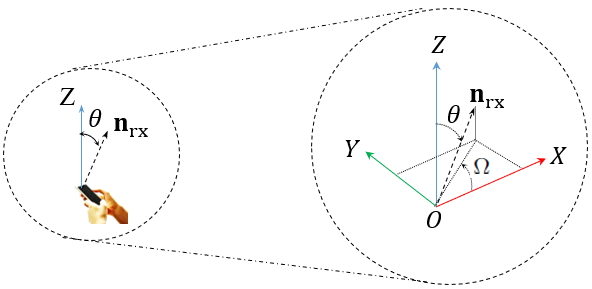}
\caption{Orientation angles of the LiFi receiver.}
\label{fig:systmodel2}
\end{figure}
\indent Based on the results of \cite{soltani2018modeling}, $\cos(\phi)$ and $\cos(\psi)$ are expressed, respectively, as
\begin{subequations}
\begin{align}
\cos(\phi) &= \frac{h_{\rm a} - h_{\rm u}}{d}, \\ 
\cos(\psi) &= \frac{\left(z_{\rm a}-z_{\rm u}\right)}{d} \cos(\theta) - \frac{\left(x_{\rm a}-x_{\rm u}\right)}{d} \cos(\Omega) \sin(\theta) \\ 
&- \frac{\left(y_{\rm a}-y_{\rm u}\right)}{d}\sin(\Omega) \sin(\theta) \nonumber,
\end{align}
\end{subequations}
where $\left(x_{\rm a},y_{\rm a},z_{\rm a}\right)$ and $\left(x_{\rm u},y_{\rm u},z_{\rm u}\right)$ are the Cartesian coordinates of the AP and the UE, respectively, and as shown in Fig.~\ref{fig:systmodel2}, $\Omega$ and $\theta$ are the angle of direction and the elevation angle of the UE, respectively. The angle of direction $\Omega$ represents the angle between the direction the user is facing and the $X$-axis, whereas the elevation angle $\theta$ is the angle between the normal vector of PD $n_{\rm rx}$ and the $Z$-axis. Based on Fig.~\ref{fig:systmodel1}, we have $\left(x_{\rm a},y_{\rm a},z_{\rm a}\right) = \left(0,0,h_{\rm a}\right)$ and $\left(x_{\rm u},y_{\rm u},z_{\rm u}\right) = \left(r\cos(\alpha),r\sin(\alpha),h_{\rm u}\right)$. Therefore, $\cos(\psi)$ can be expressed as
\begin{equation}
\cos(\psi) = \frac{r\cos(\Omega-\alpha) \sin(\theta) + \left(h_{\rm a}-h_{\rm u}\right)\cos(\theta)}{d}.
\end{equation}
Consequently, the LOS channel gain $H$ is expressed as 
    \begin{equation}
    \label{EqChannel}
    H = \left(\frac{a(\theta) r}{d^{m+3}} \cos (\Omega-\alpha) + \frac{b(\theta)}{d^{m+3}}\right) \times \mathbb{1}\left(\cos(\psi)>\cos(\Psi_c) \right), 
\end{equation}
where $a(\theta) = H_0 (h_a-h_u)^m \sin(\theta)$ and $b(\theta) = H_0 (h_a-h_u)^{m+1} \cos(\theta)$. \\
\indent Based on the above, we conclude that the random behaviour of the channel gain $H$ depends mainly on four random variables, which are $r$, $\alpha$, $\Omega$ and $\theta$. Precisely, the variables $r$ and $\alpha$ model the randomness of the instantaneous location of the LiFi receiver whereas the variables $\Omega$ and $\theta$ model the randomness of the instantaneous UE orientation. Additionally, the statistics of the polar distance $r$ and the the elevation angle $\theta$ depend on the motion of the LiFi user, either stationary or mobile. Consequently, the statistics of the LOS channel gain $H$ inducibly depend on the LiFi user activity. In the following section, the exact statistics of the channel gain $H$ are derived for the case of stationary and mobile LiFi users.
\section{Channel Statistics of Stationary and Mobile Users with Random Device orientation}
The objective of this section is deriving the exact statistics of the LOS channel gain $H$ for the case of stationary and mobile LiFi users. In subsection \ref{suc1subsecA}, we present the statistics of the four main factors $r$, $\alpha$, $\Omega$ and $\theta$ for each case, from which we derive in subsection \ref{suc1subsecB} the exact statistics of $H$. 
\subsection{Parameters Statistics}
\label{suc1subsecA}
From a statistical point of view, the instantaneous location and the instantaneous orientation of the LiFi receiver are independent. Thus, the couples of random variables $\left(r,\alpha\right)$ and $\left(\Omega,\theta \right)$ are independent. In addition, based on the results of \cite{arfaoui2018secrecy,gupta2018statistics}, the random variables $r$ and $\alpha$ are independent, since $r$ defines the polar distance and $\alpha$ defines the polar angle. On the other hand, based on the results of \cite{soltani2018modeling}, the angle of direction $\Omega$ and the elevation angle $\theta$ are also statistically independent. Therefore, the random variables $r$, $\alpha$, $\Omega$ and $\theta$ are independent. In addition, for both cases of stationary and mobile LiFi users, the random variables $\alpha$ and $\Omega$ are uniformly distributed within $[0,2\pi]$ \cite{arfaoui2018secrecy,gupta2018statistics,soltani2018modeling}. However, this is not the case for the polar distance $r$ and the elevation angle $\theta$. In fact, as we will show in the following, the statistics of $r$ and $\theta$ depend on whether the LiFi receiver is stationary or mobile.
\paragraph*{\textbf{1) Stationary Users}} \quad \\
\indent When the LiFi user is stationary, its location is fixed. However, the LiFi user is randomly located, i.e., its instantaneous location is uniformly distributed within the circular area of radius $R_{\rm e}$. In this case, the PDF of the polar distance $r$ is expressed  $f_r(r) = \frac{2r}{R_{\rm e}^2}\mathcal{U}_{\left[ 0,R_{\rm e} \right]}(r)$ \cite{arfaoui2018secrecy}. Additionally, the authors in \cite{soltani2018modeling} presented a measurement-based study for the UE orientation, where they derived statistical models for the elevation angle $\theta$. In this study, they show that, for stationary users, the elevation angle $\theta$ follows a truncated Laplace distribution, where its PDF is expressed as 
\begin{equation}
    f_{\theta} (\theta) = \frac{\exp \left( -\frac{|\theta-\mu_{\theta}|}{ \sigma_{\theta}/\sqrt{2}} \right) \mathcal{U}_{\left[ 0,\pi/2 \right]}(\theta) }{\sqrt{2\sigma_{\theta}}\left(1 - \exp \left( -\frac{\left(\frac{\pi}{2}-\mu_{\theta}\right)}{\sigma_{\theta}/\sqrt{2}} \right) - \exp \left( -\frac{\mu_{\theta}}{\sigma_{\theta}/\sqrt{2}} \right)\right)},
\end{equation}
such that $\mu_{\theta} = 41.39^\circ$ and $\sigma_{\theta} = 7.68^\circ$.
\paragraph*{\textbf{2) Mobile Users}} \quad \\
\indent For mobile users, and especially in indoor environments, the UE motion represents the user's walk, which is equivalent to a 2-D topology of the RWP mobility model, where the direction, velocity and destination points (waypoints) are all selected randomly. Based on \cite{govindan2011probability,aalo2016effect}, the spatial distribution of the LiFi receiver is polynomial in terms of the polar distance $r$ and its PDF is expressed as $f_r(r) = \sum_{i=1}^3 a_i \frac{r^{b_i}}{R_{\rm e}^{b_i+1}}\mathcal{U}_{\left[ 0,R_{\rm e} \right]}(r)$, where $[a_1,a_2,a_3] = \frac{1}{75}[324, -420, 96]$ and $[b_1, b_2, b_3] = [1,3,5]$. Moreover, it was shown in the same measurement-based study in \cite{soltani2018modeling} that, for mobile users, the elevation angle $\theta$ follows a truncated Gaussian distribution, where its PDF is expressed as 
\begin{equation}
    f_{\theta} (\theta) = \frac{2\exp \left( -\frac{(\theta-\mu_{\theta})^2}{ 2\sigma_{\theta}^2} \right)\mathcal{U}_{\left[ 0,\pi/2 \right]}(\theta)}{\sqrt{2 \pi \sigma_{\theta}^2} \left(\erf \left(\frac{\frac{\pi}{2}-\mu_{\theta}}{\sqrt{2}\sigma_{\theta}} \right) + \erf \left(\frac{\mu_{\theta}}{\sqrt{2}\sigma_{\theta}} \right) \right)},
\end{equation}
such that $\mu_{\theta} = 29.67^\circ$ and  $\sigma_{\theta} = 7.78^\circ$.
\subsection{Channel Statistics}
\label{suc1subsecB}
As stated in Section II, the LiFi receiver can be located anywhere inside the outer cell with radius $R_{\rm e}$. However, it is connected to the desired AP if it is located inside the LiFi attocell, i.e., if $r \in [0,R]$. In other words, in order to have a communication link between the desired AP and the LiFi receiver, the only admitted values of the polar distance $r$ should be within the range $[0,R]$. Due to this, we constrain the range of $r$ to be $[0,R]$, and therefore, the exact PDF of the polar distance $r$ becomes $
    \tilde{f}_r(r) = \frac{f_r(r)}{F_r(R) - F_r(0)}\mathcal{U}_{\left[ 0,R \right]}(r)$,
where $F_r$ denotes the CDF of $r$. Consequently, the PDF of the distance $d = \sqrt{r^2 + (h_a-h_u)^2}$ is given by
\begin{equation}
    f_d (d) = \frac{d \times \tilde{f}_r \left( \sqrt{d^2-(h_a-h_u)^2} \right)}{\sqrt{d^2-(h_a-h_u)^2}}  \mathcal{U}_{\left[d_{\min}, d_{\max} \right]}(d), 
\end{equation}
where $d_{\min} = h_a-h_u$ and $d_{\max} = \sqrt{R^2+(h_a-h_u)^2}$. On the other hand, consider the random variable $\cos \left(\Omega-\alpha \right)$  appearing in \eqref{EqChannel}. Since $\Omega$ and $\alpha$ are independent and uniformly distribution within $[0,2\pi]$ and using the PDF transformation of random variables, $\cos \left(\Omega-\alpha \right)$ follows the arcsine distribution within the range $[-1,1]$. Thus, the PDF and CDF of $\cos \left(\Omega-\alpha \right)$ are expressed, respectively, as
\begin{equation}
f_{\cos (\Omega-\alpha)}(x) = \frac{1}{\pi \sqrt{1-x^2}} \mathcal{U}_{\left[-1,1 \right]}(x) ,
\end{equation}
\begin{equation}
F_{\cos (\Omega-\alpha)}(x) = \left(\frac{\arcsin (x)}{\pi} + \frac{1}{2}\right)\mathcal{U}_{\left[-1,1 \right]}(x) + \mathcal{U}_{\left[1,+\infty \right]}(x).
\end{equation}
Based on this, the exact PDF of the channel gain $H$ is given in the following theorem.
\begin{theorem}
The range of the LOS channel gain $H$ is $[h_{\min},h_{\max}]$, where $h_{\min}=0$ and $h_{\max}= \frac{H_0}{\left(h_a-h_u \right)^{2}}$. In addition, for $h \in [h_{\min}^{},h_{\max}]$, the PDF of $H$ is expressed as
\begin{equation}
    f_{H} (h) = g_{H}(h)\mathcal{U}_{\left[h_{\min}^{*},h_{\max} \right]}(h) + F_{\cos(\psi)}(\cos(\Psi_c)) \delta(h),
\end{equation}
where $h_{\min}^{*} = \frac{H_0(h_a-h_u)^{m}\cos(\Psi_c)}{d_{\max}^{m+2}}$, $F_{\cos(\psi)}(\cos(\Psi_c))$ is given in (13) on top of this page,
\begin{figure*}[t] 
\begin{equation}
F_{\cos(\psi)}(\cos(\Psi_c)) =  \int_{d_{\min}}^{d_{\max}} \int_{0}^{\frac{\pi}{2}} F_{\cos \left(\Omega-\alpha \right)} \left(\frac{d \cos(\Psi_c)-(h_a-h_u) \cos \theta}{\sin (\theta) \sqrt{{d}^2-(h_a-h_u)^2}} \right) f_{\theta}(\theta) \mathrm{d}\theta f_{d}(d) \mathrm{d}d 
\end{equation}
\noindent\makebox[\linewidth]{\rule{\textwidth}{0.4pt}}
\end{figure*} 
in which $d_{\min}^*(h) = \max \left(d_{0}(h),d_{\min} \right)$ such that
\begin{equation}
    d_{0}(h) = \left(\frac{h_0(h_a-h_u)^{m}\cos(\Psi_c)}{h}\right)^{\frac{1}{m+2}}.
\end{equation}
and the function $g_{H}$ is expressed as shown in (15) on top of this page,
\begin{figure*}[t] 
\begin{equation}
\begin{aligned}
g_{H} (h) &=  \int_{d_{\min}^*(h)}^{d_{\max}} \int_{0}^{\frac{\pi}{2}} \frac{d^{m+3}}{a(\theta)\sqrt{d^2-(h_a-h_u)^2}} f_{\cos (\Omega-\alpha)} \left( \frac{d^{m+3}h-b(\theta)}{a(\theta)\sqrt{d^2-(h_a-h_u)^2}} \right) f_{\theta}(\theta) f_{d}(d) \mathrm{d}\theta \mathrm{d}d \\ 
&+v(h) \int_{0}^{\frac{\pi}{2}} J_{H} \left(\theta,d \right) f_{\theta}(\theta) f_{d}(d) \mathrm{d}\theta \mathrm{d}d,
\end{aligned}
\end{equation}
\noindent\makebox[\linewidth]{\rule{\textwidth}{0.4pt}}
\end{figure*} 
in which the function $v$ is expressed as $v(h) = \frac{-\left( h_0(h_a-h_u)^{m}\cos(\Psi_c) \right)^{\frac{1}{m+2}}}{(m+2)h^{\frac{m+3}{m+2}}} \mathcal{U}_{\left[ h_{\min}^{*},h_{\max}^{*} \right]}$, such that $h_{\max}^{*} = \frac{H_0 (h_a-h_u)^m\cos(\Psi_c)}{d_{\min}^{m+2}}$, and the function $J_{H}$ is expressed as
\begin{equation}
    \begin{split}
        J_{H} \left(\theta,d \right) &= F_{\cos (\Omega-\alpha)} \left(\frac{d_{\min}^* \cos(\Psi_c)-(h_a-h_u) \cos (\theta)}{\sin(\theta)\sqrt{{d_{\min}^*}^2-(h_a-h_u)^2}} \right)\\ 
        &\quad - F_{\cos (\Omega-\alpha)} \left( \frac{{d_{\min}^*}^{m+3}h-b(\theta)}{a(\theta)\sqrt{{d_{\min}^*}^2-(h_a-h_u)^2}} \right).
    \end{split}
\end{equation}
\end{theorem}
\begin{proof}
See Appendix A.
\end{proof}
\noindent The exact CDF of the LOS channel gain $H$ is also provided in (40) in Appendix A. On the other hand, note that the function $h \mapsto F_{\cos(\psi)}(\cos(\Psi_c)) \delta(h)$ expresses the effect of the field of view $\Psi_c$ on the LOS channel gain $H$. \\ 
\indent As it can be seen in Theorem 1, the closed-form expression of the exact PDF of the LOS channel gain $H$ in (12) is neither straightforward nor tractable, since it involves some complex and atypical integrals. Due to this, in order to provide simple and tractable channel models for indoor LiFi systems, we propose in the following section some approximations for the PDF of $H$ in (12), for the cases of stationary and mobile LiFi users.
\section{Approximate PDFs of LiFi LOS channel Gain}
In this section, our objective is to derive some approximations for the PDF of $H$, starting from the results of Theorem 1. The cases of stationary and mobile LiFi users are investigated separately in subsections \ref{sub4A} and \ref{sub4B}, respectively.
\subsection{Stationary Users}
\label{sub4A}
An approximate expression of the PDF of the LOS channel gain $H$ for the case of stationary LiFi users is given in the following theorem. 
\begin{theorem}
For the case of stationary users, an approximate expression of the PDF of the channel gain $H$ is given by 
\begin{equation}
    f_{H} (h) \approx \frac{1}{h^\nu} g (h) + F_{\cos(\psi)}(\cos(\Psi_c)) \delta(h),
\end{equation}
where $\nu >0$ and $g$ is a function with range $[h_{\min}^{*},h_{\max}]$. 
\end{theorem}
\begin{proof}
See Appendix B.
\end{proof}
\noindent The approximation of the PDF of the LOS channel gain $H$ provided in Theorem 2 expresses two main factors, which are the random location and random orientation of the UE. The functions $h \mapsto \frac{1}{h^\nu}$ and $h \mapsto g (h)$ express respectively the effects of the random location of the receiver and the random orientation of the UE on the LOS channel gain $H$. At this point, the missing part is the function $g$ that provides the best approximation for the PDF of the LOS channel gain $f_{H}$. In the following, we provide two approximate expressions for the PDF $g$.
\paragraph*{1) The Modified Truncated Laplace (MTL) Model} \quad \\ 
\indent Since the function $h \mapsto g (h)$ expresses the effect of the random orientation of UE on the LOS channel gain $H$ and motivated by the fact that the elevation angle $\theta$ follows a truncated Laplace distribution as shown in (7), one reasonable choice for $g$ is the Laplace distribution. Consequently, an approximate expression of the PDF of the LOS channel gain $H$ can be given by 
\begin{equation}
\begin{split}
    f_{H} (h) &\approx \frac{h^{-\nu} \exp \left( -\frac{|h-\mu_{H}|}{ b_{H}} \right)}{M_1\left(-\nu, \mu_H, b_H \right)} \mathcal{U}_{\left[ h_{\min}^{*},h_{\max} \right]}(h) \\ 
    &+ F_{\cos(\psi)}(\cos(\Psi_c)) \delta(h),
\end{split}
\end{equation}
where $\mu_{H} \in [h_{\min}^{*},h_{\max}]$, $b_{H}>0$ and $M_1\left(-\nu, \mu_H, b_H \right)$ is a normalization factor given by 
\begin{equation}
M_1\left(-\nu, \mu_H, b_H \right) = \frac{G_1\left(-\nu, \mu_H, b_H \right)}{\left[1-F_{\cos(\psi)}(\cos(\Psi_c)) \right]},
\end{equation}
where $G_1$ is given in (20) on top of this page,
\begin{figure*}[t] 
\begin{equation}
G_1 \left(\gamma,\mu_H, b_H \right) = -b_H^{1+\gamma} e^{\frac{-\mu_H}{b_H}} \left[ \Gamma \left(1 + \gamma, \frac{h_{\max}}{b_H} \right) - \Gamma \left(1 + \gamma, \frac{\mu_{H}}{b_H} \right)  + (-1)^{1-\gamma} \left( \Gamma \left(1 + \gamma, -\frac{\mu_{H}}{b_H} \right) - \Gamma \left(1 + \gamma, -\frac{h_{\min}^{*}}{b_H} \right) \right) \right],
\end{equation}
\noindent\makebox[\linewidth]{\rule{\textwidth}{0.4pt}}
\end{figure*} 
in which $\Gamma$ denotes the upper incomplete Gamma function. At this stage, we need to determine the parameters $\left(\nu,\mu_H,b_H\right)$ of $f_H$. One approach to do this is through moments matching. Using the exact PDF of $H$ in (16), the non-centered moments of the LOS channel gain $H$ are given by 
\begin{equation}
m_i^e = \int_{h_{\min}^{*}}^{h_{\max}} h^{i} g_H(h) \mathrm{d}h + F_{\cos(\psi)}(\cos(\Psi_c)), \quad i \in \mathbb{N},
\end{equation}
whereas by using the approximate PDF of $H$ in (18), the non-centered moments of the LOS channel gain $H$ are given by  
\begin{equation}
m_i^a\left(\nu,\mu_H,b_H\right) =   \frac{M_1\left(i-\nu, \mu_H, b_H \right)}{M_1\left(-\nu, \mu_H, b_H \right)}, \quad i \in \mathbb{N},
\end{equation}
Therefore, since only three parameters need to be determined, which are $\left(\nu,\mu_H,b_H\right)$, they can be obtained by solving the following system of equations
\begin{equation}
    m_i^a\left(\nu,\mu_H,b_H\right) = m_i^e, \quad \text{for} \, \, i=1,2,3.
\end{equation}
\paragraph*{2) The Modified Beta (MB) Model} \quad \\
\label{MBmodel}
\indent The exact PDF of the LOS channel gain $H$ involves the integral of a function that has the form $(x,y) \mapsto f_{\cos (\Omega-\alpha)}(g(x,y))$. Since $\cos \left(\Omega - \alpha \right)$ follows the arcsine distribution and based on the fact that the arcsine distribution is a special case of the Beta distribution, we approximate the function $g$ with a Beta distribution. Consequently, an approximate expression of the PDF of the LOS channel gain $H$ can be given by
\begin{equation}
\begin{split}
    &f_{H} (h) \\ 
    &\approx \frac{h^{-\nu} \left(\frac{h-h_{\min}^{*}}{h_{\max}-h_{\min}^{*}}\right)^{\alpha_H-1} \left(\frac{h_{\max} - h}{h_{\max}-h_{\min}^{*}}\right)^{\beta_H-1}}{M_2\left(-\nu, \alpha_H, \beta_H \right)}\mathcal{U}_{\left[ h_{\min}^{*},h_{\max} \right]}(h) \\ 
    &+ F_{\cos(\psi)}(\cos(\Psi_c)) \delta(h),
\end{split}
\end{equation}
where $\alpha_H>0$, $\beta_{H}>0$ and $M_2\left(-\nu, \alpha_H, \beta_H \right)$ is a normalization factor given by 
\begin{equation}
M_2\left(-\nu, \alpha_H, \beta_H \right) = \frac{G_2\left(-\nu, \alpha_H, \beta_H \right)}{\left[1-F_{\cos(\psi)}(\cos(\Psi_c)) \right]},
\end{equation}
such that $G_2$ is given in (26) at the top of the next page, 
\begin{figure*}[t] 
\begin{equation}
\begin{split}
G_2 \left(\gamma,\alpha_H, \beta_H \right) = B(\gamma,\alpha_H,\beta_H) \times &\left[ \Gamma (\beta_H) \Gamma (-\gamma) h_{\max}^{\alpha_H+\gamma} {}_{2}\tilde{F}_1\left(1-\alpha_H,-\alpha_H-\beta_H-\gamma+1;-\alpha_H-v+1;\frac{h_{\min}^{*}}{h_{\max}}\right) \right. \\ &\quad \left. -\Gamma (\alpha_H) h_{\min}^{{*}^{\alpha_H+\gamma}} \Gamma (\alpha_H+\beta_H+\gamma) {}_{2}\tilde{F}_1\left(1-\beta_H,\gamma+1;\beta_H+\gamma+1;\frac{h_{\min}^{*}}{h_{\max}}\right)\right],
\end{split}
\end{equation}
\noindent\makebox[\linewidth]{\rule{\textwidth}{0.4pt}}
\end{figure*} 
in which ${}_{2}\tilde{F}_1$ denotes the regularized hyper-geometric function and
\begin{equation}
B(\gamma,\alpha_H,\beta_H) = \frac{\pi  h_{\max}^{\beta_H-1} (h_{\max}-h_{\min}^{*})^{-\alpha_H-\beta_H+2}}{\sin (\pi  (\alpha_H+\gamma))\Gamma (-\gamma) \Gamma (\alpha_H+\beta_H+\gamma)}.
\end{equation}
Based on the above, it remains to derive the parameters $\left(\nu,\alpha_H,\beta_H\right)$ of $f_H$. Similar to the case of the MTL model, one approach to do this is through moments matching. Specifically, $\left(\nu,\alpha_H,\beta_H\right)$ can be obtained by solving the system of equations in (23), where for $i=1,2,3$, $m_i^a$ is expressed in this case as
\begin{equation}
m_i^a\left(\nu,\mu_H,b_H\right) = \frac{M_2\left(i-\nu, \alpha_H, \beta_H \right)}{M_2\left(-\nu, \alpha_H, \beta_H \right)}, \quad i \in \mathbb{N}.
\end{equation}
\subsection{Mobile Users}
\label{sub4B}
An approximate expression of the PDF of the LOS channel gain $H$ for the case of mobile LiFi users is given in the following theorem. 
\begin{theorem}
For the case of mobile users, an approximate expression of the PDF of the channel gain $H$ is given by 
\begin{equation}
    f_{H} (h) \approx \sum_{j=1}^3 \frac{1}{h^{\nu_i}} g_{j} (h) + F_{\cos(\psi)}(\cos(\Psi_c)) \delta(h),
\end{equation}
where, for $j=1,2,3$, $\nu_j >0$ and $g_{j}$ is a function with range $[h_{\min}^{*},h_{\max}]$.
\end{theorem}
\begin{proof}
See Appendix C.
\end{proof}
It is important to highlight here that, for $j=1,2,3$, the functions $h \mapsto \frac{1}{h^{\nu_j}}$ and $h \mapsto g_{j} (h)$ express respectively the effects of user mobility and the random orientation of the UE on the LOS channel gain $H$. At this point, the missing part is the functions $g_{j}$, for $j=1,2,3$, that provide the best approximation for the PDF of the LOS channel gain $f_{H}$. In the following, we provide two expressions for each function $g_{j}$ for $j=1,2,3$.
\paragraph*{1) The Sum of Modified Truncated Gaussian (SMTG) Model} \quad \\ 
\indent Since for $j=1,2,3$, the functions $h \mapsto g_{j} (h)$ express the effect of the random orientation of the UE on the channel gain $H$ and motivated by the fact that, for the case of mobile LiFi users, the elevation angle $\theta$ follows a truncated Gaussian distribution as shown in (8), one reasonable choice for the functions $g_{j}$ is the truncated Gaussian distribution. Consequently, an approximate expression of the PDF of the LOS channel gain $H$ can be given by
\begin{equation}
\begin{split}
    f_{H} (h) &\approx \frac{\sum_{j=1}^3 h^{-\nu_j} \exp \left( -\frac{\left(h-\mu_{H,j}\right)^2}{ 2 \sigma_{H,j}^2} \right)}{\sum_{j=1}^3 M_3\left(-\nu_j, \mu_{H,j}, \sigma_{H,j} \right)} \mathcal{U}_{\left[ h_{\min}^{*},h_{\max} \right]}(h) \\ 
    &+ F_{\cos(\psi)}(\cos(\Psi_c)) \delta(h),
\end{split}
\end{equation}
where for $j=1,2,3$, $\mu_{H,j} \in [h_{\min}^{*},h_{\max}]$, $\sigma_{H,j}>0$ and $M_3\left(-\nu_j, \mu_{H,j}, \sigma_{H,j} \right)$ is a normalization factor that is given by 
\begin{equation}
M_3\left(-\nu_j, \mu_{H,j}, \sigma_{H,j} \right) = \frac{\int_{h_{\min}^{*}}^{h_{\max}} h^{-\nu_j} \exp \left( -\frac{\left(h-\mu_{H,j}\right)^2}{ 2 \sigma_{H,j}^2} \right) \mathrm{d}h}{\left[1-F_{\cos(\psi)}(\cos(\Psi_c)) \right]}.
\end{equation}
Now, in order to to have the complete closed-form expression of $f_H$, we have to determine the parameters $\left\{ \left(\nu_j,\mu_{H,j},b_{H,j}\right), j=1,2,3 \right\}$. Similar to the one of the stationary users case, one approach to determine these parameters is through moments matching. Specifically, since only nine parameters need to be determined, which are $\left\{\left(\nu_j,\mu_{H,j},\sigma_{H,j}\right) \left| j=1,2,3 \right. \right\}$, they can be obtained by solving the following system of equations
\begin{equation}
    m_i^a = m_i^e, \quad \text{for} \, \, i=1,2,..,9,
\end{equation}
where, for $i=1,2,...,9$, $m_i^a$ is expressed in this case as
\begin{equation}
m_i^a = \frac{\sum_{j=1}^3 M_3\left(i-\nu_j, \mu_{H,j}, b_{H,j} \right)}{\sum_{j=1}^3 M_3\left(-\nu_j, \mu_{H,j}, \sigma_{H,j} \right)}.
\end{equation}
\paragraph*{2) The Sum of Modified Beta (SMB) Model} \quad \\ 
\indent Motivated by the same reasons as for the MB model in Section IV-A1, we approximate each function $g_{j}$, for $ j=1,2,3$, with a Beta distribution. Consequently, an approximate expression of the PDF of the LOS channel gain $H$ is given in (34) at the top of this page, where $\alpha_{H,j}>0$, $\beta_{H,j}>0$ and $M_2\left(-\nu_j, \alpha_{H,j}, \beta_{H,j} \right)$ is given in (25).
\begin{figure*}[t] 
\begin{equation}
f_{H} (h) \approx \frac{\sum_{j=1}^3 h^{-\nu_j} \left(\frac{h-h_{\min}^{*}}{h_{\max}-h_{\min}^{*}}\right)^{\alpha_{H,j}-1} \left(\frac{h_{\max} - h}{h_{\max}-h_{\min}^{*}}\right)^{\beta_{H,j}-1}\mathcal{U}_{\left[ h_{\min}^{*},h_{\max} \right]}(h)}{\sum_{j=1}^3 M_2\left(-\nu_j, \alpha_{H,j}, \beta_{H,j} \right)} + F_{\cos(\psi)}(\cos(\Psi_c)) \delta(h),
\end{equation}
\noindent\makebox[\linewidth]{\rule{\textwidth}{0.4pt}}
\end{figure*} 
Finally, it remains now to derive the parameters $\left\{\left(\nu_j,\alpha_{H,j},\beta_{H,j}\right) \left| j=1,2,3 \right. \right\}$ of $f_H$. Similar to the STMG model, these parameters can be obtained by solving the by solving the system of equations in (32), where for $i=1,2,...,9$, $m_i^a$ is expressed in this case as
\begin{equation}
m_i^a = \frac{\sum_{j=1}^3 M_2\left(i-\nu_j, \alpha_{H,j}, \beta_{H,j} \right)}{\sum_{j=1}^3 M_2\left(-\nu, \alpha_H, \beta_H \right)}.
\end{equation}
\subsection{Summary of the Proposed Models}
\subsection{Summary of the Proposed Models}
A detailed algorithm for implementing the proposed statistical channel models for indoor LiFi systems is presented in Algorithm 1.
\begin{algorithm}[t]
\caption{Detailed algorithm for implementing the proposed statistical channel models}
\begin{algorithmic} 
\STATE 1. \textbf{Input}: 
\begin{enumerate}[label=\roman*)]
\item Attocell dimensions $\left(R, h_{\rm a}\right)$.
\item AP's parameters $\left(\rho, \phi_{1/2}\right)$.
\item UE's height $h_{\rm u}$
\item UE's parameters $\left(R_{\rm p},n_{\rm c}, A_{\rm g},\Psi_c \right)$.
\end{enumerate}
\STATE 2. \textbf{Calculate} $H_0$ as shown in (3).
\STATE 3. \textbf{Calculate} $F_{\cos(\psi)}(\cos(\Psi_c))$ as shown in (13).
\STATE 4. \textbf{If} the LiFi user is stationary:
\begin{enumerate}[label=\roman*)]
\item MTL model:
\begin{enumerate}[label=\alph*)]
\item Estimate the parameters using (23).
\item Inject the parameters into the PDF in (18).
\end{enumerate}
\item MB model:
\begin{enumerate}[label=\alph*)]
\item Estimate the parameters using (23).
\item Inject the parameters into the PDF in (24).
\end{enumerate}
\end{enumerate}
\STATE \quad \textbf{elseif} the LiFi user is mobile:
\begin{enumerate}[label=\roman*)]
\item SMTG model:
\begin{enumerate}[label=\alph*)]
\item Estimate the parameters using (32).
\item Inject the parameters into the PDF in (30).
\end{enumerate}
\item SMB model:
\begin{enumerate}[label=\alph*)]
\item Estimate the parameters using (32).
\item Inject the parameters into the PDF in (35).
\end{enumerate}
\end{enumerate}
\STATE \quad \textbf{end}
\end{algorithmic}
\end{algorithm} 
\section{Simulation Results and Discussions}
In this paper, we consider a typical indoor LiFi attocell \cite{arfaoui2019SNR,gupta2018statistics}. Parameters used throughout the paper are shown in Table \ref{T2}. 
\begin{table}[t]
\caption{Simulation Parameters}
\centering
\renewcommand{\arraystretch}{1.3} 
\setlength{\tabcolsep}{0.18cm} 
\begin{tabular}{| c | c | c |}
  \hline 
  Parameter & Symbol & Value \\
  \hline
  Ceiling height & $h_a$ & $2.4$ m \\ 
  \hline
  LED half-power semiangle & $\phi_{1/2}$ & $60^\circ$\\ 
  \hline 
  LED conversion factor & $\rho$ & $0.7$ W/A \\
  \hline
  PD responsivity & $R_p$ & $0.6$ A/W \\ 
  \hline
  PD geometric area & $A_g$ & $1$ cm$^2$ \\ 
  \hline 
  Optical concentrator refractive index & $n_c$ & 1 \\
  \hline 
  UE's height (stationary) & $h_u$ & $0.9$ m \\ 
  \hline
  UE's height (mobile) & $h_u$ & $1.4$ m \\
  \hline 
\end{tabular} 
\label{T2}
\end{table}
In Subsection \ref{secV1}, we present the PDF and CDF of the LOS channel gain $H$ for the case of stationary and mobile LiFi users. In Subsection \ref{secV2}, we investigate the error performance of Indoor LiFi systems using the derived statistics of the LOS channel gain $H$. Finally, based on the error performance presented in \ref{secV2}, we propose in subsection \ref{secV3} an optimized design for the indoor cellular system that can enhance the performance of LiFi systems.
\subsection{Channel Statistics}
\label{secV1}
For stationary LiFi users, Figs. \ref{fig:S1} and \ref{fig:S2} present the theoretical, simulated and approximated PDF and CDF of the LOS channel gain $H$, when a radius of the attocell of $R=1$m and $R=2.5$m, respectively. For both cases, two different values for the field of view of the UE were considered, which are $\Psi_c = 90^\circ$ and $60^\circ$. 
\begin{figure}[t]
\centering     
\subfigure[PDF: $\Psi_c =  90^\circ$.]{\label{fig:S1a}\includegraphics[width=0.51\linewidth]{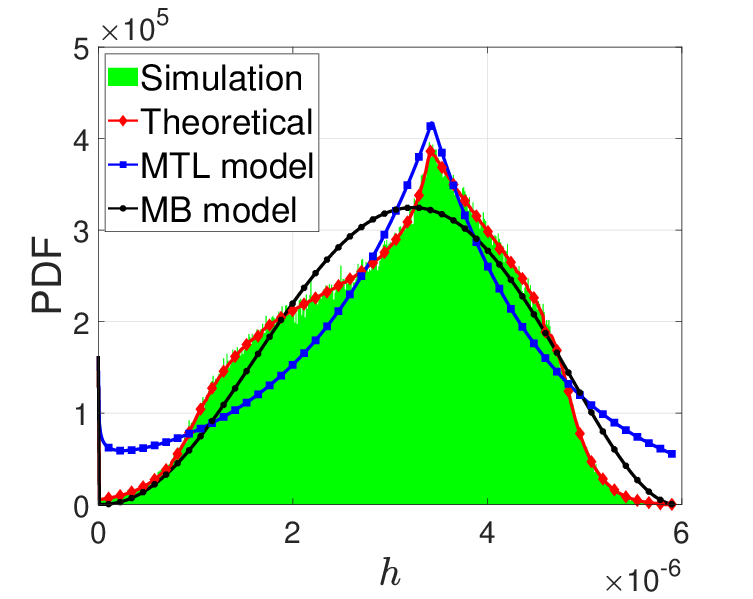}}
\subfigure[CDF: $\Psi_c =  90^\circ$.]{\label{fig:S1b}\includegraphics[width=0.51\linewidth]{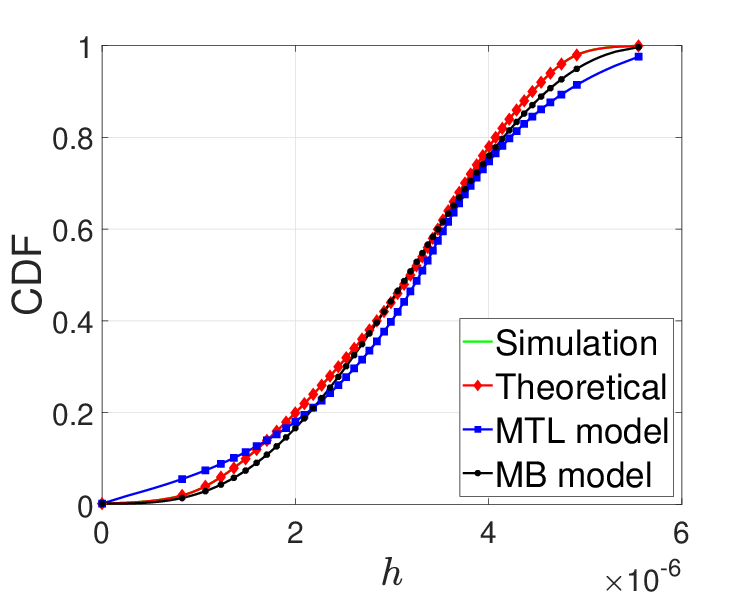}}
\subfigure[PDF: $\Psi_c =  60^\circ$.]{\label{fig:S1c}\includegraphics[width=0.51\linewidth]{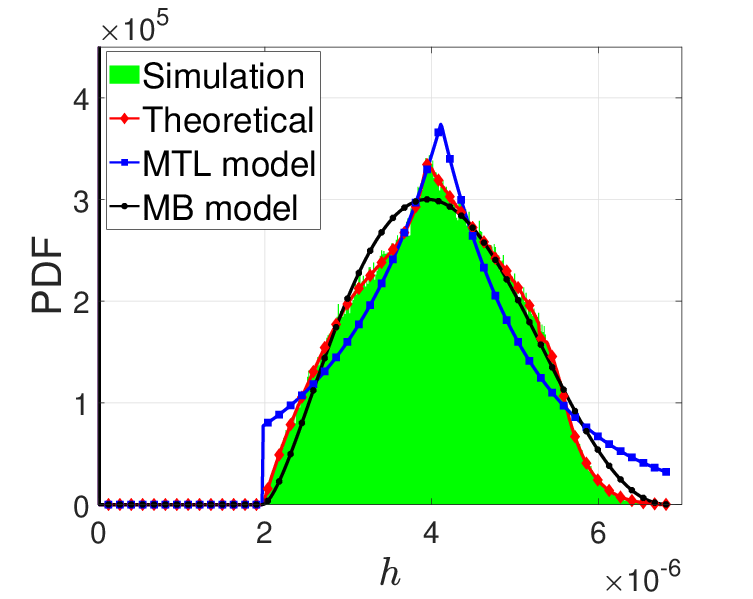}}
\subfigure[CDF: $\Psi_c =  60^\circ$.]{\label{fig:S1d}\includegraphics[width=0.51\linewidth]{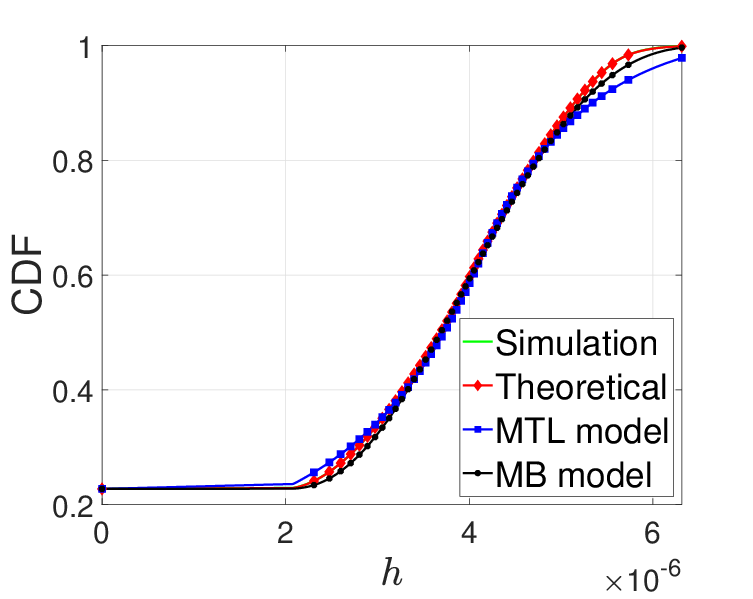}}
\caption{Comparison between the simulation, theoretical and approximation results of the PDF and the CDF of the LOS channel gain $H$ for the case of stationary LiFi users when $R = 1$m.}
\label{fig:S1}
\end{figure}
\begin{figure}[t]
\centering     
\subfigure[PDF: $\Psi_c =  90^\circ$.]{\label{fig:S2a}\includegraphics[width=0.51\linewidth]{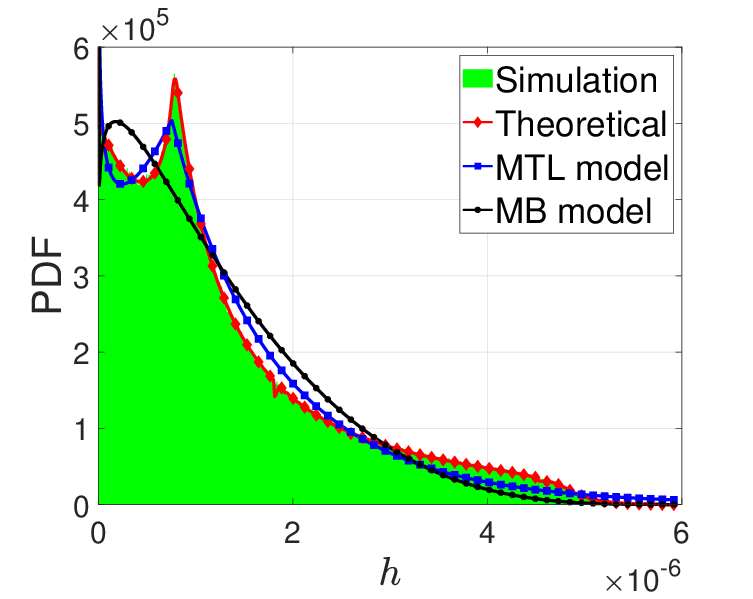}}
\subfigure[CDF: $\Psi_c =  90^\circ$.]{\label{fig:S2b}\includegraphics[width=0.51\linewidth]{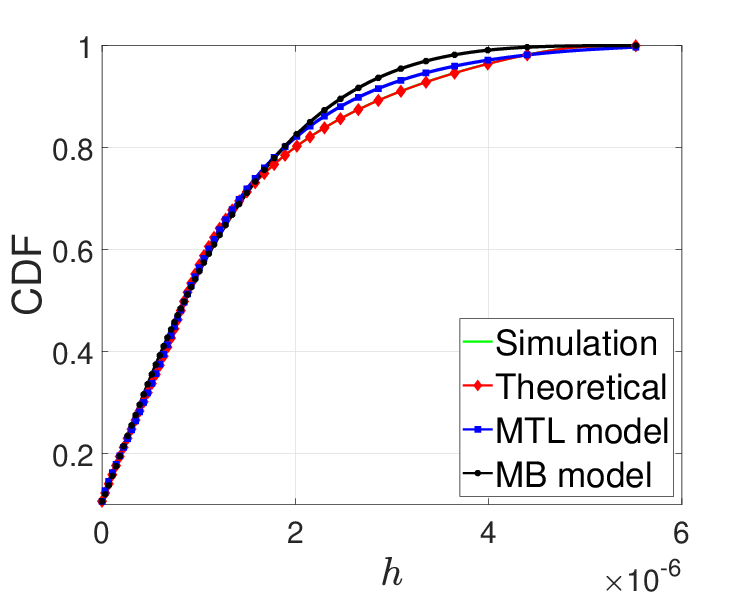}}
\subfigure[PDF: $\Psi_c =  60^\circ$.]{\label{fig:S2c}\includegraphics[width=0.51\linewidth]{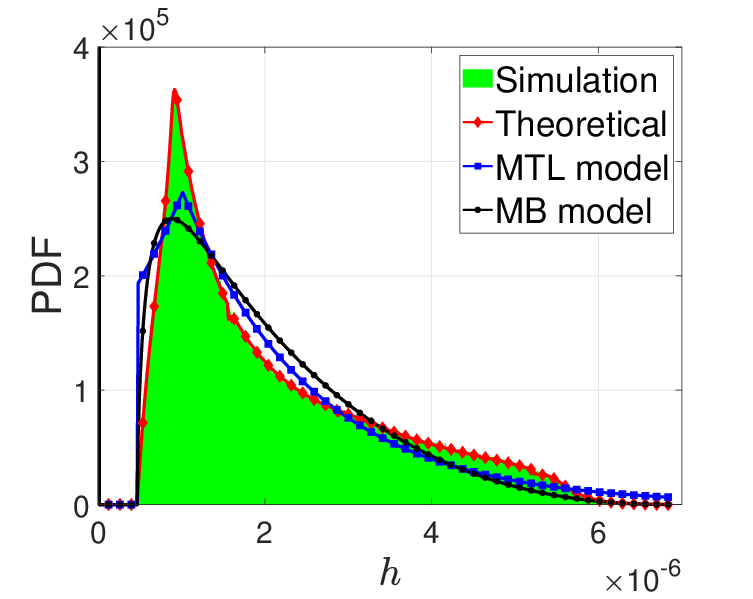}}
\subfigure[CDF: $\Psi_c =  60^\circ$.]{\label{fig:S2d}\includegraphics[width=0.51\linewidth]{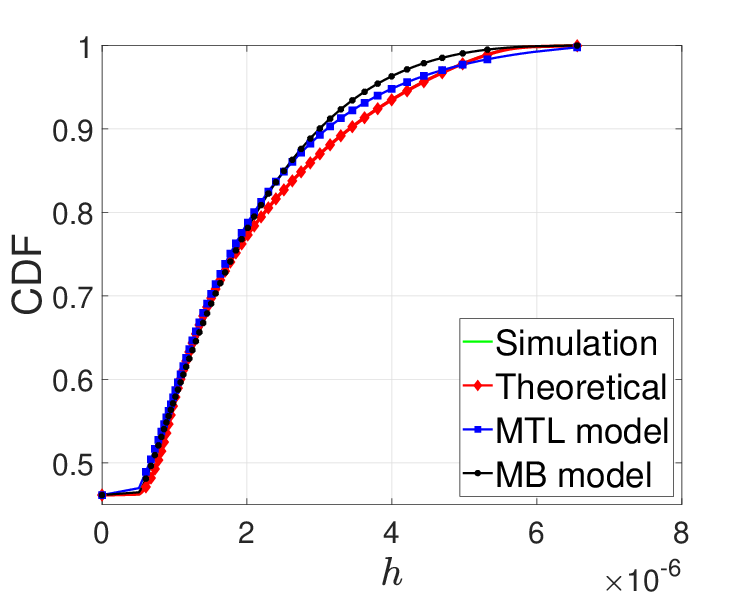}}
\caption{Comparison between the simulation, theoretical and approximation results of the PDF and the CDF of the LOS channel gain $H$ for the case of stationary LiFi users when $R = 2.5$m.}
\label{fig:S2}
\end{figure}
These figures show that the proposed MTL and MB models offer good approximation for the distribution of the LOS channel gain $H$. Analytically, in order to evaluate the goodness of the proposed MTL and MB models, we use the well-known Kolmogorov-Smirnov distance (KSD) \cite{massey1951kolmogorov}. In fact, the KSD measures the absolute distance between two distinct CDFs $F_1$ and $F_2$ \cite{massey1951kolmogorov}, i.e., 
\begin{equation}
    \text{KSD} = \underset{x}{\max} \left|F_1(x)-F_2(x) \right|.
\end{equation}
Obviously, smaller values of KSD correspond to more similarity between distributions. In our case, the KSD of the MTL and MB models are shown in Table \ref{T3}.
\begin{table}[t]
\caption{KSD of MTL and MB models}
\begin{center}
\renewcommand{\arraystretch}{1.3} 
\setlength{\tabcolsep}{0.18cm} 
\begin{tabular}{| l | c | c | c | c |}
    \cline{2-5} 
  \multicolumn{1}{c|}{}  & \multicolumn{2}{c|}{MTL model} & \multicolumn{2}{c|}{MB model} \\
  \cline{2-5} 
  \multicolumn{1}{c|}{}  & $\Psi_c = 90^\circ$ & $\Psi_c = 60^\circ$ & $\Psi_c = 90^\circ$ & $\Psi_c = 60^\circ$ \\
  \hline
  $R=1$m & 0.0669 & 0.0448 & 0.0336 & 0.0197\\ 
  \hline
  $R=2.5$m & 0.0239 & 0.0241 & 0.0444 & 0.0316\\ 
  \hline 
\end{tabular} 
\end{center}
\label{T3}
\end{table}
As it can be seen in this table, the maximum KSD value for the MTL and MB models are 0.0669 and 0.0444, respectively, which demonstrates the good approximation offered by the MTL and MB models. \\ 
\indent For mobile LiFi users, Figs. \ref{fig:S3} and \ref{fig:S4} present the theoretical, simulated and approximated PDF and CDF of the LOS channel gain $H$, when the radius of the attocell is $R=1$m and $R=2.5$m, respectively. For both cases, two different values for the field of view of the UE were considered, which are $\Psi_c = 90^\circ$ and $60^\circ$. 
\begin{figure}[t]
\centering     
\subfigure[PDF: $\Psi_c =  90^\circ$.]{\label{fig:S3a}\includegraphics[width=0.51\linewidth]{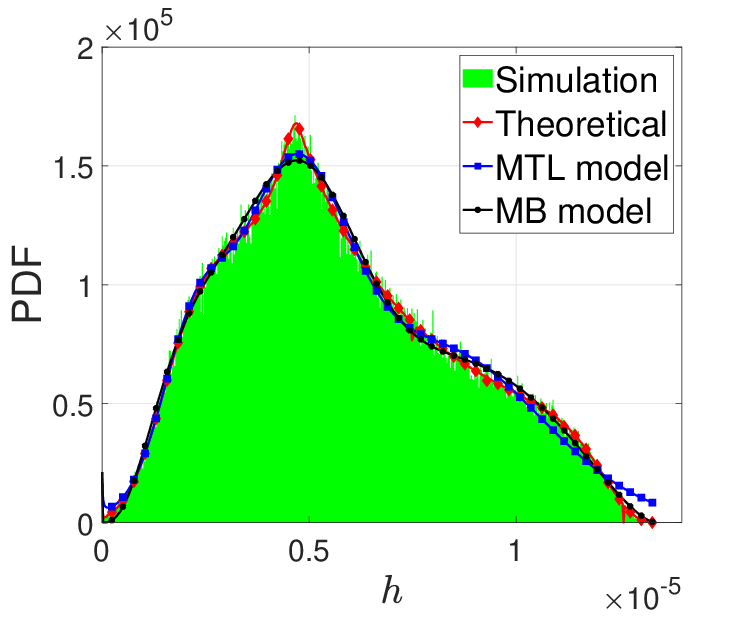}}
\subfigure[CDF: $\Psi_c =  90^\circ$.]{\label{fig:S3b}\includegraphics[width=0.51\linewidth]{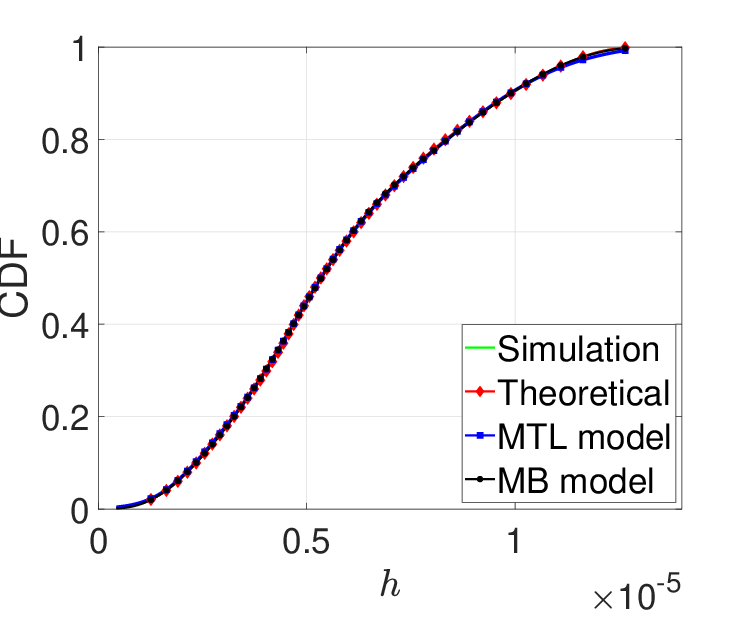}}
\subfigure[PDF: $\Psi_c =  60^\circ$.]{\label{fig:S3c}\includegraphics[width=0.51\linewidth]{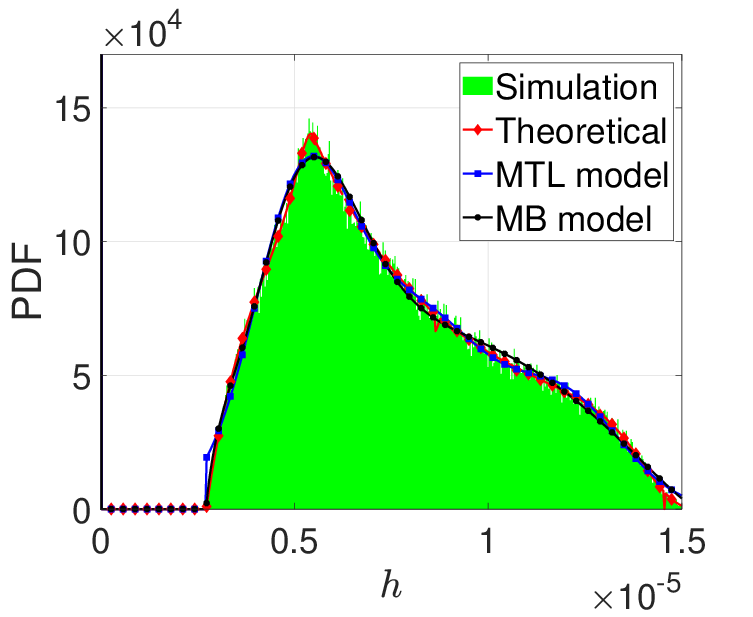}}
\subfigure[CDF: $\Psi_c =  60^\circ$.]{\label{fig:S3d}\includegraphics[width=0.51\linewidth]{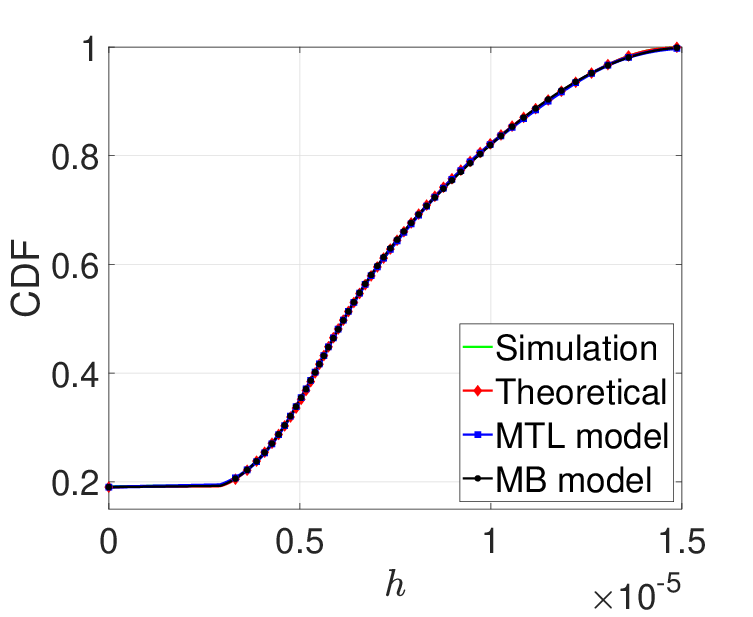}}
\caption{Comparison between the simulation, theoretical and approximation results of the PDF and the CDF of the LOS channel gain $H$ for the case of mobile LiFi users when $R = 2.5$m.}
\label{fig:S3}
\end{figure}
\begin{figure}[t]
\centering     
\subfigure[PDF: $\Psi_c =  90^\circ$.]{\label{fig:S4e}\includegraphics[width=0.51\linewidth]{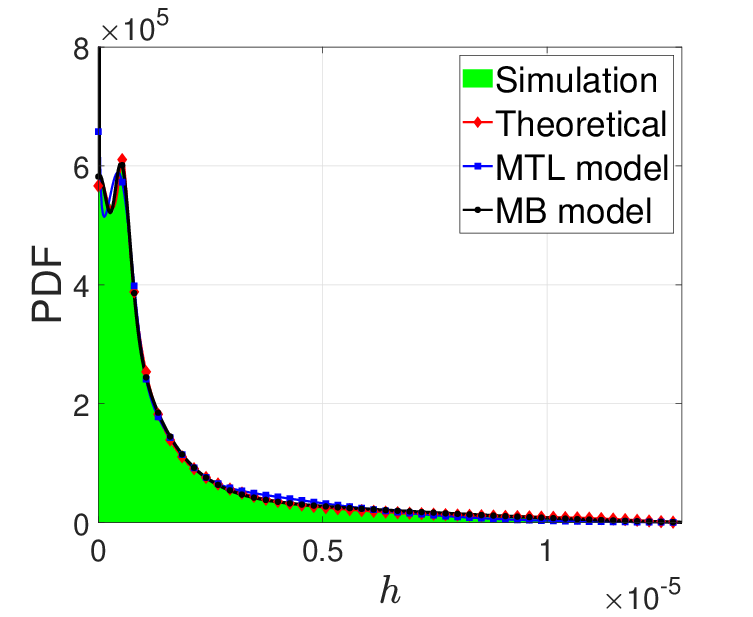}}
\subfigure[CDF: $\Psi_c =  90^\circ$.]{\label{fig:S4f}\includegraphics[width=0.51\linewidth]{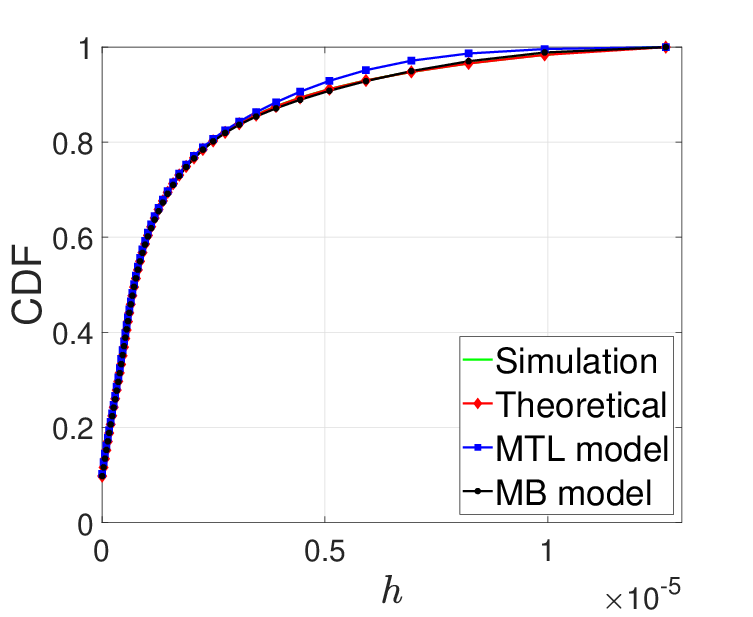}}
\subfigure[PDF: $\Psi_c =  60^\circ$.]{\label{fig:S4g}\includegraphics[width=0.51\linewidth]{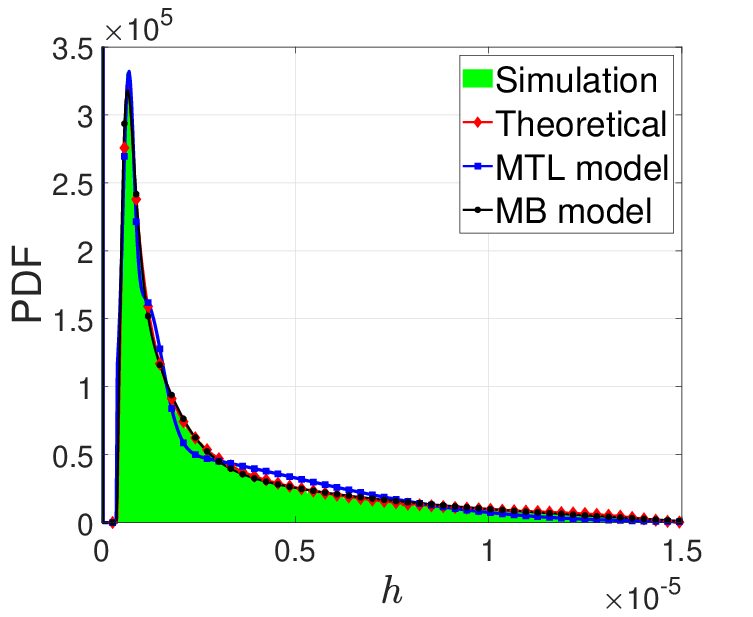}}
\subfigure[CDF: $\Psi_c =  60^\circ$.]{\label{fig:S4h}\includegraphics[width=0.51\linewidth]{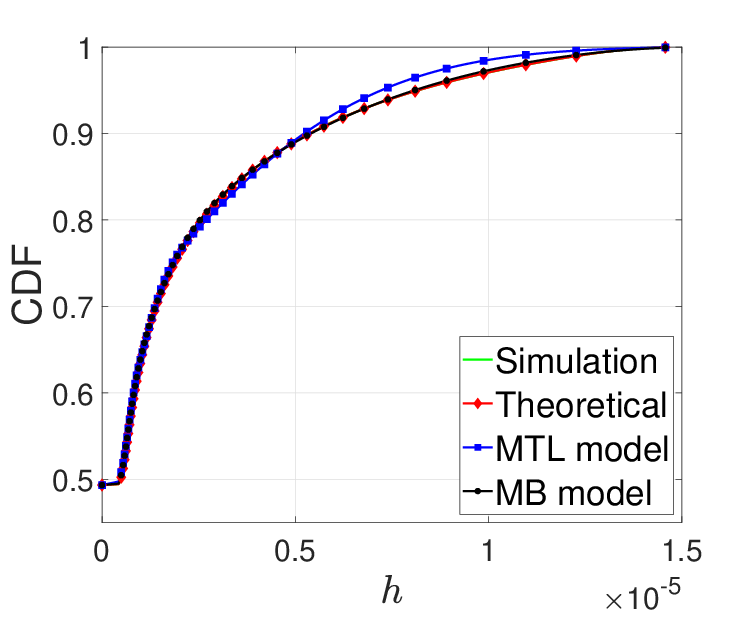}}
\caption{Comparison between the simulation, theoretical and approximation results of the PDF and the CDF of the LOS channel gain $H$ for the case of mobile LiFi users when $R = 2.5$m.}
\label{fig:S4}
\end{figure}
These figures show that the proposed SMTG and SMB models offer good approximation for the distribution of the LOS channel gain $H$. Furthermore, using the KSD metric, Table \ref{T4} presents the KSD of the SMTG and the SMB models, where it shows that their maximum KSD values are 0.0238 and 0.0054, respectively. This result demonstrates the good approximation offered by the MTL and MB models.
\begin{table}[t]
\caption{KSD of SMTG and SMB models}
\begin{center}
\renewcommand{\arraystretch}{1.3} 
\setlength{\tabcolsep}{0.18cm} 
\begin{tabular}{| l | c | c | c | c |}
    \cline{2-5} 
  \multicolumn{1}{c|}{}  & \multicolumn{2}{c|}{MTL model} & \multicolumn{2}{c|}{MB model} \\
  \cline{2-5} 
  \multicolumn{1}{c|}{}  & $\Psi_c = 90^\circ$ & $\Psi_c = 60^\circ$ & $\Psi_c = 90^\circ$ & $\Psi_c = 60^\circ$ \\
  \hline
  $R=1$m & 0.0082 & 0.0037 & 0.0048 & 0.0030\\ 
  \hline
  $R=2.5$m & 0.0238 & 0.0156 & 0.0054 & 0.0047\\ 
  \hline 
\end{tabular} 
\end{center}
\label{T4}
\end{table}
\subsection{Error Performance}
\label{secV2}
\begin{figure}[t]
\centering     
\subfigure[Stationary user, $\Psi_c =  90$ .]{\label{fig:S5a}\includegraphics[width=0.525\linewidth]{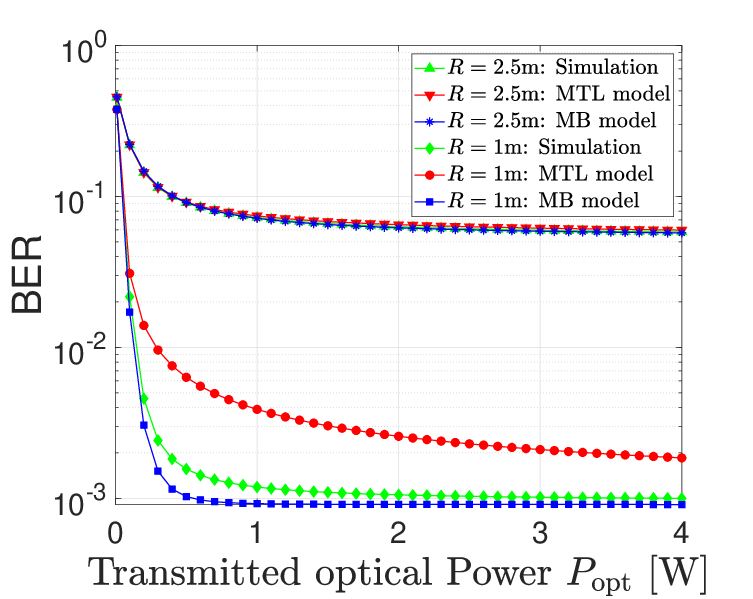}}
\subfigure[Stationary user, $\Psi_c =  60^\circ$.]{\label{fig:S5b}\includegraphics[width=0.525\linewidth]{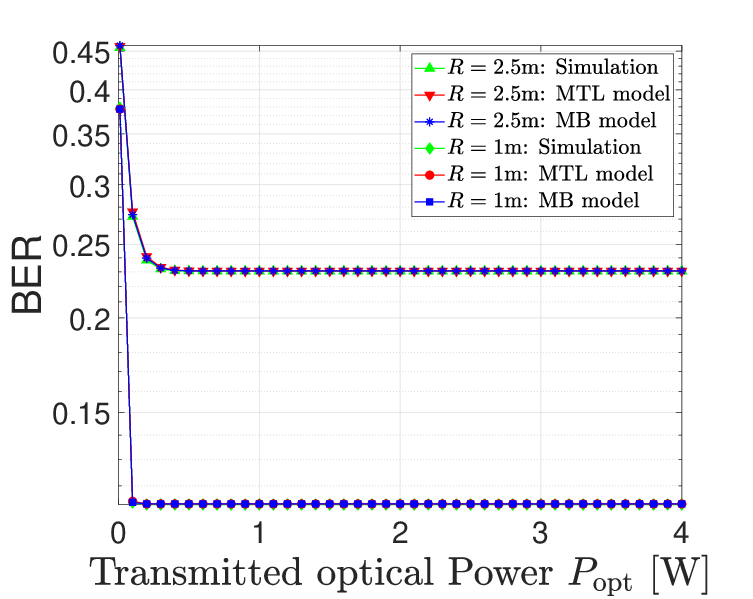}}
\subfigure[Mobile user, $\Psi_c =  90^\circ$.]{\label{fig:S5c}\includegraphics[width=0.525\linewidth]{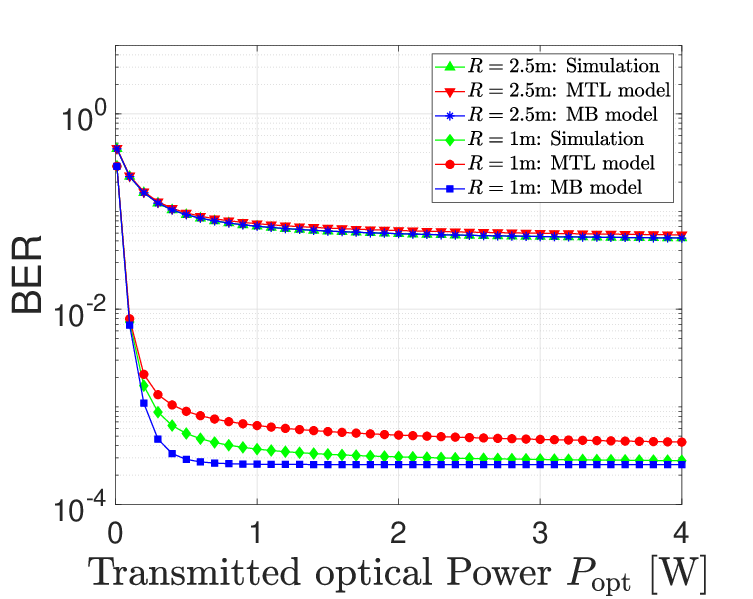}}
\subfigure[Mobile user, $\Psi_c =  60^\circ$.]{\label{fig:S5d}\includegraphics[width=0.525\linewidth]{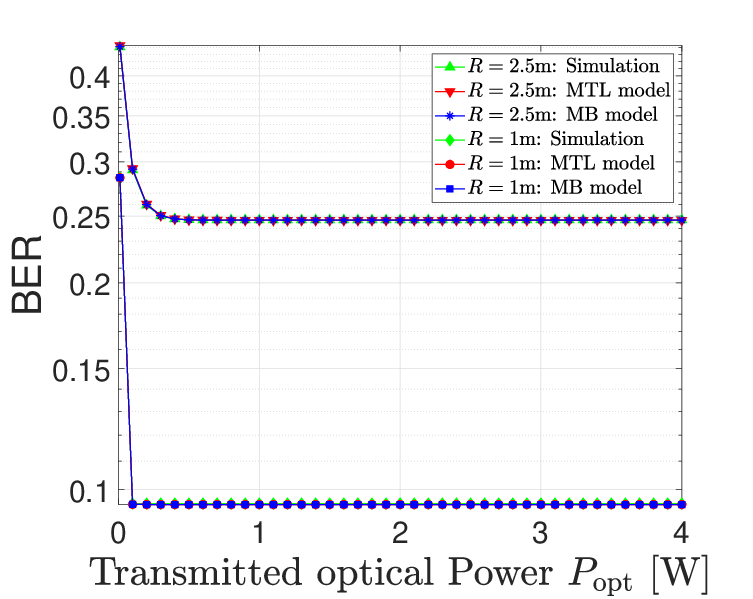}}
\caption{BER performance of OOK modulation versus the transmitted optical power for stationary and mobile LiFi users.}
\label{fig:S5}
\end{figure}
Fig. \ref{fig:S5} presents the average BER performance of the on-off keying (OOK) modulation versus the transmitted optical power $P_{\rm opt}$ for stationary and mobile users for the values of the attocell radius of $R=1$m and $2.5$m and for the values of the field of view $\Psi_c=90^\circ$ and $\Psi_c=60^\circ$. Considering stationary users, this figure shows that the BER results of the MTL and MB models match perfectly the simulated BER for both cases when $\left(R,\Psi_c \right) = \left(2.5\text{m},90^\circ \right)$ and when $\Psi_c = 60^\circ$. However, for the case when $\left(R,\Psi_c \right) = \left(1\text{m},90^\circ \right)$, we remark that the average BER results of the MB model match the simulated BER better than the ones of the MTL model. This can be also seen from the values of the KSD in Table \ref{T3}, where we can see that the KSD of the MB model is lower than the one of the MTL model when $\left(R,\Psi_c \right) = \left(1\text{m},90^\circ \right)$. In other words, when $\left(R,\Psi_c \right) = \left(1\text{m},90^\circ \right)$, the MB model offers better accuracy than the MTL model. This is mainly due to the assumptions made for both models. In fact, when the radius of the attocell $R$ is small and by referring to (5), the random variable $\cos \left(\Omega-\alpha \right)$ is dominant in $\cos \left(\psi \right)$. Hence, assuming that the distribution of the random orientation of the UE can be approximated by a Beta distribution makes more sense. For mobile users, the same figure shows that the BER results of the SMTG and the SMB models match perfectly the simulated BER for both cases when $\left(R,\Psi_c \right) = \left(2.5\text{m},90^\circ \right)$ and when $\Psi_c = 60^\circ$. However, for the case when $\left(R,\Psi_c \right) = \left(1\text{m},90^\circ \right)$, we remark that the BER results of the SMB model matches the simulated BER better than the ones of the SMTG model. Similar to the case of stationary users, the SMB model offers better accuracy than the SMTG model when $\left(R,\Psi_c \right) = \left(1\text{m},90^\circ \right)$ due to the assumptions made for both models.  \\ 
\indent Fig. \ref{fig:S5} shows also two important facts about the BER performance of LiFi users. First, it can be seen that the BER performance degrades heavily when either the radius of the attocell $R$ increases or the field of view of the LiFi receiver decreases. Second, the BER saturates as the transmitted optical power increases. These two facts can be explained by the following corollary.  \\
\begin{corollary}
At high transmitted optical power $P_{\rm opt}$, the average probability of error of the $M$-ary pulse amplitude modulation (PAM) for the considered LiFi system is given by 
\begin{equation}
    \lim_{P_{\rm opt} \rightarrow \infty} P_e \left( P_{\rm opt} \right) = \frac{F_{\cos(\psi)}(\cos(\Psi_c))}{2}.
\end{equation}
\end{corollary}
\begin{proof}
See appendix D.
\end{proof}
The result of Corollary 1 shows that, even when the transmitted optical power $P_{\rm opt}$ is high, the BER is stagnating at $\frac{F_{\cos(\psi)}(\cos(\Psi_c))}{2}$. This result is directly related to the cases when the AP is out of the FOV of the LiFi receiver. On the other hand, based on its expression in (19), $F_{\cos(\psi)}(\cos(\Psi_c))$ is a function of the attocell radius $R$ and the field of view of the receiver $\Psi_c$. Therefore, since $d_{\max}$ increases as $R$ increases, then $F_{\cos(\psi)}(\cos(\Psi_c))$ is an increasing function in $R$. In addition, since $x \mapsto F_{\cos(\psi)}(x)$ is a CDF, it is an increasing function, and due to the fact that $x \mapsto \cos (x)$ is a decreasing function within $[0,\pi/2]$, then $F_{\cos(\psi)}(\cos(\Psi_c))$ increases as $\Psi_c$ decreases. The aforementioned reasons explain the bad BER performance of the LiFi system when either the radius of the attocell $R$ increases or the field of view $\Psi_c$ decreases. \\ 
\indent From a practical point of view, the above performance can be explained as follows. Recall that 
\begin{subequations}
\begin{align}
    F_{\cos(\psi)}(\cos(\Psi_c)) &= {\rm Pr} \left( \cos(\psi) \leq \cos(\Psi_c) \right) \\ 
    &= {\rm Pr} \left( H \leq 0 \right),
\end{align}
\end{subequations}
which is literally the outage probability of the LiFi system, i.e., the probability that the UE is not connected to the AP even when it is inside the attocell. Obviously, for large values of $R$ or small values of $\Psi_c$, the probability that the LiFi receiver is not connected to the AP increases. This is mainly due to the effects the random location of the LiFi user along with the random orientation of the UE and it explains the bad BER performance in this case. \\
\indent The question that may come to mind here is how can one enhance the performance of the LiFi system under such a realistic environment? Recently, some practical solutions have been proposed in the literature to alleviate the effects of the random behaviour of LiFi channel. These solutions include the use of MIMO LiFi systems along with transceiver designs that have high spatial diversity gains such as the multidirectional receiver (MDR) \cite{mohammad2018optical,ImanICCW19}, the omnidirectional transceiver \cite{chen2018omnidirectional} and the angular diversity transceiver \cite{chen2018reduction}. In the following subsection, we propose a new design of indoor LiFi MIMO systems that can alleviate the effects of the random location of the LiFi user along with the random orientation of the UE.
\subsection{Design Consideration of Indoor LiFi Systems}
\label{secV3}
\indent The concept of optical MIMO systems has been introduced in practical LiFi systems, where multiple LiFi APs cooperate together and serve multiple users within the resulting illuminated area \cite{fath2012performance,tavakkolnia2018energy,tavakkolnia2019mimo,soltani2018bidirectional}. Each LiFi AP creates an optical attocell and the respective illumination areas of the adjacent attocells overlap with each other. Consider the indoor LiFi MIMO system shown in Fig. \ref{fig:S6}, which consists of five APs that correspond to small and adjacent attocells, where each has radius $R_c$. The distance between the AP of the attocell in the middle (green attocell), which we refer to as the reference attocell, and the APs of the remaining adjacent attocells is $D_c$.  \\ 
\indent Let us assume that a LiFi user is located within the reference attocell, where all five APs serve this user by transmitting the same signal. One way to reduce the outage probability of the LiFi user, i.e., the probability that it is not connected to one of the APs, is through a well designed attocells radius $R_c$ and APs spacing $D_c$ that guarantee a maximum target probability of error $P_e^{\rm th}$, without any handover protocol or coordination scheme between the different APs. Fig. \ref{fig:S7} presents the BER performance of a LiFi user that is located within the reference attocell, where the field of view of the UE is $\Psi_c = 60^\circ$. Both stationary and mobile cases are considered and different values of $R_c$ and $D_c$ are evaluated. By comparing the results of this figure and those of Fig. \ref{fig:S5}, for the case when $R=1$m for example, we can see how the coexisting APs can significantly improve the BER performance of the system. In addition, we remark from Fig. \ref{fig:S7} that the choice of $\left(R_c,D_c \right)$ has also a big impact on the BER performance, for example, for the case of a stationary user, the best choice among the considered values is $\left(R_c,D_c \right) = \left(1\text{m},1.5\text{m} \right)$, whereas for the case of a mobile user, the best choice is $\left(R_c,D_c \right) = \left(1\text{m},1\text{m} \right)$. Overall, for a target probability of error $P_e^{\rm th} = 3.8 \times 10^{-3}$, we conclude that the choice $\left(R_c,D_c \right) = \left(1\text{m},1\text{m} \right)$ is the best choice that guarantees the target performance jointly for both stationary and mobile users. Obviously, the optimal $\left(R_c,D_c \right)$ depends on the geometry of the attocells and the parameters of the UE as well, such as the height of the AP $h_{\rm a}$ and the height of the UE $h_{\rm u}$. This problem will be investigated in future works.
\begin{figure}[t]
\centering     
\subfigure[Top view.]{\label{fig:S6a}\includegraphics[width=0.88\linewidth]{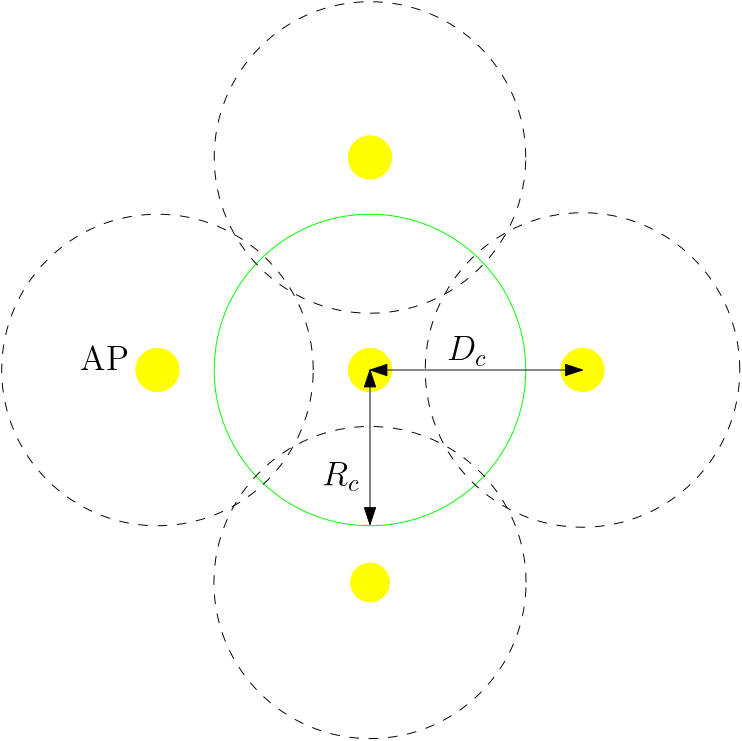}}
\subfigure[Horizontal view.]{\label{fig:S6b}\includegraphics[width=0.88\linewidth]{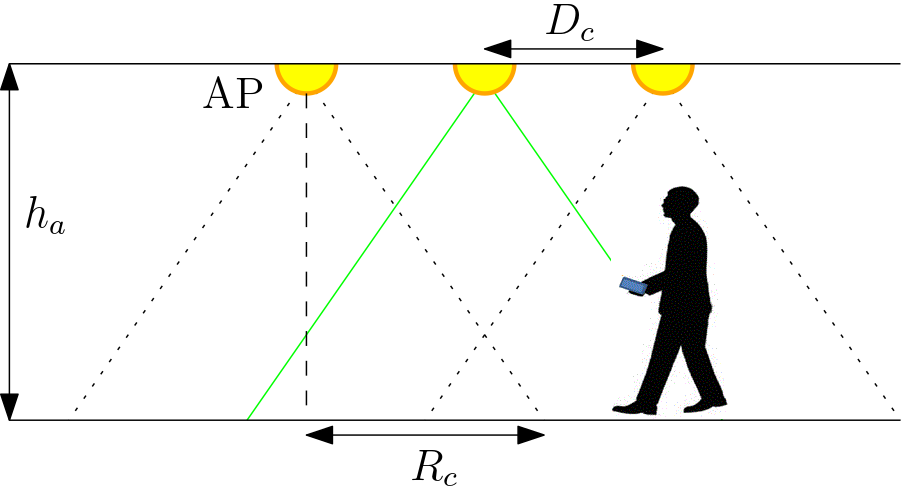}}
\caption{An indoor multi-cell LiFi system.}
\label{fig:S6}
\end{figure}
\begin{figure}[t]
\centering     
\subfigure[Stationary user.]{\label{fig:S7a}\includegraphics[width=1\linewidth]{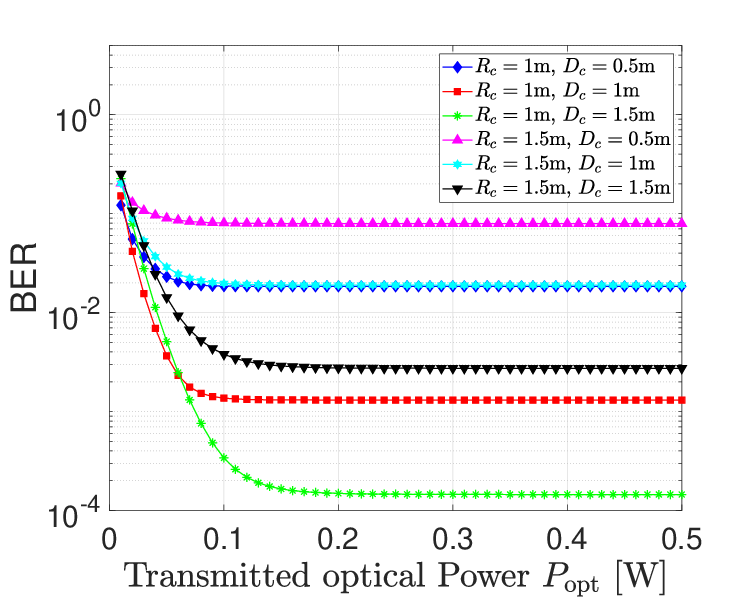}}
\subfigure[Mobile user.]{\label{fig:S7b}\includegraphics[width=1\linewidth]{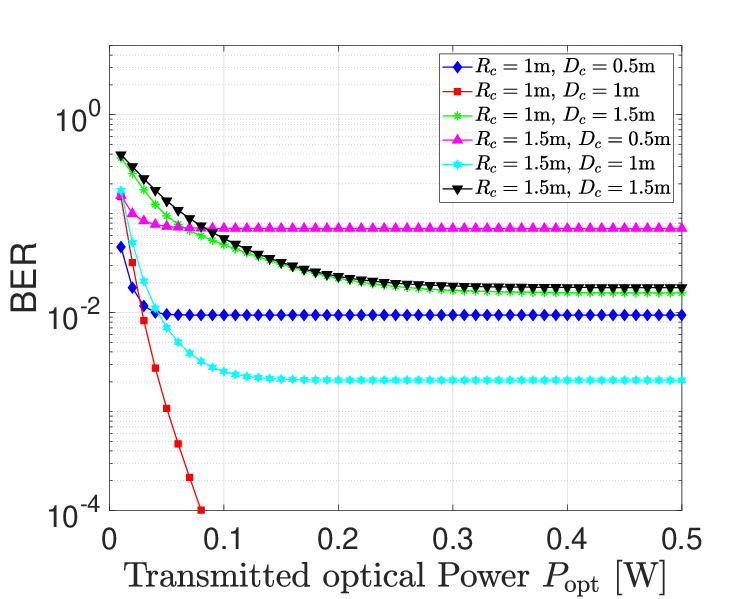}}
\caption{BER performance of OOK modulation versus the transmitted optical power for stationary and mobile LiFi users when $\Psi_c = 60^\circ$.}
\label{fig:S7}
\end{figure}
\section{Conclusions and Future Works}
In this paper, novel, realistic, and measurement-based channel models for indoor LiFi systems have been proposed. The statistics of the LOS channel gain are derived for the case of stationary and mobile LiFi users, where the LiFi receiver is assumed to be randomly oriented. For stationary LiFi users, the MTL and the MB models were proposed, whereas for the case of mobile users, the SMTG and SMB models were proposed. The accuracy of each model was evaluated using the KSD. In addition, the effect of random orientation and spatial distribution of LiFi users on the error performance of LiFi users was investigated based on the derived models. Our results showed that the random behaviour and motion of LiFi users has strong effect on the LOS channel gain. Therefore, we proposed a novel design of indoor LiFi MIMO systems in order to guarantee the required reliability performance for reliable communication links. \\ 
\indent The channel models proposed in this paper, albeit being fundamental and original, they serve as a starting point for developing realistic transmission techniques and transceiver designs tailored to real-world set-ups in an effort to bring the deployment of LiFi systems closer than ever. Thus, investigating optimal transceiver designs and cellular architectures based on the derived channel models that can meet the high demands of 5G and beyond in realistic communication environment can be considered as a future research direction. In addition, the derived channel models are intended for downlink communication in indoor LiFi environments. Therefore, deriving similar models for uplink transmission and for outdoor environments should also be considered in future work.
\section*{Acknowledgment}
M. A. Arfaoui and C. Assi acknowledge the financial support from Concordia University and FQRNT. M. D. Soltani acknowledges the School of Engineering for providing financial support. A. Ghrayeb is supported in part by Qatar National Research Fund under NPRP Grant NPRP8-052-2-029 and in part by FQRNT. H. Haas acknowledges the financial support from the Wolfson Foundation and Royal Society. He also gratefully acknowledges financial support by the Engineering and Physical Sciences Research Council (EPSRC).
\appendices 
\section{Proof of Theorem 1}
At first, let us determine the range of the LOS channel gain $H$. Recall that $H$ is expressed as 
\begin{equation}
    H = H_0  \frac{\left(h_a-h_u \right)^{m}\cos(\psi)}{d^{m+2}} \times \mathbb{1}\left(\cos(\psi)>\cos(\Psi_c) \right),
\end{equation}
where $
\cos(\psi) = \frac{r\cos(\Omega-\alpha) \sin(\theta) + \left(h_{\rm a}-h_{\rm u}\right)\cos(\theta)}{d}$. Since $\Psi_c \in [0,\frac{\pi}{2}]$, we have $H \geq 0$ and the minimum value of $H$ is equal to $h_{\min}=0$. A set of values that can yield in $h_{\min}=0$ is given as $d = d_{\min}$, $\Omega -\alpha = \pm k\frac{\pi}{2}$ s.t. $k=1,3$ and $\theta \geq \Psi_c$, which corresponds to the case where $\cos \left(\psi \right) \leq \cos \left( \Psi_c \right)$. On the other hand, the maximum value of $H$ is equal to $h_{\max}= \frac{H_0}{\left(h_a-h_u \right)^{2}}$. A set of values that can yield in $h_{\max}$ is expressed as $d = h_{\rm a}-h_{\rm u}$, $\Omega - \alpha= \pm \frac{\pi}{2}$ s.t. $k=1,3$ and $\theta = 0$, which corresponds to the case where $\cos \left(\psi \right) = 0$. On the other hand, the CDF of the LOS channel gain is expressed as shown in equation (40) on top of next page,
\begin{figure*}[t] 
\begin{subequations}
\begin{align}
F_{H} (h) &= {\rm Pr} \left(H \leq h \right) \\ 
&= {\rm Pr} \left(\left(a(\theta)\frac{\sqrt{d^2-(h_a-h_u)^2}}{d^{m+3}} \cos (\Omega-\alpha) + \frac{b(\theta)}{d^{m+3}}\right) \times \mathbb{1}\left(\cos(\psi)>\cos(\Psi_c) \right) \leq h \right) \\ 
&= {\rm Pr} \left(a(\theta)\frac{\sqrt{d^2-(h_a-h_u)^2}}{d^{m+3}} \cos (\Omega-\alpha) + \frac{b(\theta)}{d^{m+3}} \leq h, 0 \leq \psi \leq \Psi_c \right) + {\rm Pr} \left(0 \leq h, \Psi_c \leq \psi \right) \\ 
&= {\rm Pr} \left(a(\theta)\frac{\sqrt{d^2-(h_a-h_u)^2}}{d^{m+3}} \cos (\Omega-\alpha) + \frac{b(\theta)}{d^{m+3}} \leq h, \cos(\Psi_c) \leq \cos(\psi)   \right) + {\rm Pr} \left(0 \leq h, \cos(\psi) \leq \cos(\Psi_c) \right) \\ 
&= \int_{d_{\min}}^{d_{\max}} \int_{0}^{\frac{\pi}{2}}  {\rm Pr} \left( \frac{d \cos(\Psi_c)-b^{'}(\theta)}{a^{'}(\theta)\sqrt{d^2-(h_a-h_u)^2}} \leq \cos (\Omega-\alpha) \leq \frac{d^{m+3}h-b}{a(\theta)\sqrt{d^2-(h_a-h_u)^2}} \bigg|  \theta, d  \right) f_{\theta}(\theta) f_{d}(d) \mathrm{d}\theta \mathrm{d}d \\ 
&\quad+ F_{\cos(\psi)}(\cos(\Psi_c)) \mathcal{U}_{[0,\infty]}(h) \nonumber  \\ 
&= \int_{d_{\min}^*(h)}^{d_{\max}} \int_{0}^{\frac{\pi}{2}} I_{H} \left(\theta,d \right) f_{\theta}(\theta) f_{d}(d) \mathrm{d}\theta \mathrm{d}d + F_{\cos(\psi)}(\cos(\Psi_c)) \mathcal{U}_{[0,\infty]}(h) ,
\end{align}
\end{subequations}
\noindent\makebox[\linewidth]{\rule{\textwidth}{0.4pt}}
\end{figure*} 
where $a^{'}(\theta) = \sin (\theta)$, $b^{'}(\theta) = (h_a-h_u) \cos (\theta)$ and the function $I_{H}$ is expressed as
\begin{equation}
\begin{split}
   I_{H} \left(\theta,d \right) &= F_{\cos (\Omega-\alpha)} \left( \frac{d^{m+3}h-b(\theta)}{a(\theta)\sqrt{d^2-(h_a-h_u)^2}} \right) \\ 
   &- F_{\cos (\Omega-\alpha)} \left(\frac{d \cos(\Psi_c)-b^{'}(\theta)}{a^{'}(\theta)\sqrt{d^2-(h_a-h_u)^2}} \right).  
\end{split}
\end{equation}
Equality (40g) follows from the fact that the conditional probability under the integral in (40f) is not null if and only if 
\begin{equation}
\frac{d \cos(\Psi_c)-b^{'}(\theta)}{a^{'}(\theta)\sqrt{d^2-(h_a-h_u)^2}} \leq \frac{d^{m+3}h-b(\theta)}{a(\theta)\sqrt{d^2-(h_a-h_u)^2}}, 
\end{equation}
i.e., 
\begin{equation}
    \left(\frac{H_0(h_a-h_u)^{m}\cos_{\Psi_c}}{h}\right)^{\frac{1}{m+2}} \leq d.
\end{equation}
Or, $d$ is constrained within the range $[d_{\min},d_{\max}]$. Therefore, the conditional probability in (40e) is not null if and only if $d \in \left[d_{\min}^*(h), d_{\max}\right]$, where $d_{\min}^*(h) = \max \left(d_{0}(h),d_{\min} \right)$ such that $d_{0}(h) = \left(\frac{H_0(h_a-h_u)^{m}\cos(\Psi_c)}{h}\right)^{\frac{1}{m+2}}$. Additionally, since $d_0$ should be always lower than $d_{\max}$, we conclude that the channel gain $h$ under the integral in (50h) should satisfy $h_{\min}^{*} \leq h$, where 
\begin{equation}
    h_{\min}^{*} = \frac{H_0(h_a-h_u)^{m}\cos(\Psi_c)}{d_{\max}^{m+2}} \in \left[h_{\min}, h_{\max} \right].
\end{equation}
Furthermore, $F_{\cos(\psi)}(\cos(\Psi_c))$ in (40) is given as shown in (41) at the top of the next page. 
\begin{figure*}[t] 
\begin{subequations}
\begin{align}
F_{\cos(\psi)}(\cos(\Psi_c)) &= {\rm Pr} \left( \cos(\psi) \leq \cos(\Psi_c) \right) \\
&= \int_{d_{\min}}^{d_{\max}} \int_{0}^{\frac{\pi}{2}} {\rm Pr} \left( \cos(\psi) \leq \cos(\Psi_c) \bigg| \theta, d \right)  f_{\theta}(\theta) f_{d}(d) \mathrm{d}\theta \mathrm{d}d \\
&= \int_{d_{\min}}^{d_{\max}} \int_{0}^{\frac{\pi}{2}} {\rm Pr} \left( \cos (\Omega-\alpha) \leq \frac{d \cos(\Psi_c)-b^{'}(\theta)}{a^{'}(\theta)\sqrt{d^2-(h_a-h_u)^2}} \bigg| \theta, d \right) f_{\theta}(\theta) f_{d}(d) \mathrm{d}\theta \mathrm{d}d \\
&=\int_{d_{\min}}^{d_{\max}} \int_{0}^{\frac{\pi}{2}} F_{\cos (\Omega-\alpha)} \left(\frac{d \cos(\Psi_c)-(h_a-h_u) \cos (\theta)}{\sin (\theta) \sqrt{{d}^2-(h_a-h_u)^2}} \right) f_{\theta}(\theta) f_{d}(d) \mathrm{d}\theta \mathrm{d}d.
\end{align}
\end{subequations}
\noindent\makebox[\linewidth]{\rule{\textwidth}{0.4pt}}
\end{figure*} 
\indent Based on this, the corresponding PDF of the LOS channel gain $H$ is obtained by differentiating equality (40g) with respect to $h$. Using Leibniz integral rule for differentiation, the PDF of the LOS channel gain $H$ is expressed as shown in (46) on the top of next page. 
\begin{figure*}[t] 
\begin{subequations}
\begin{align}
f_{H} (h) &= \frac{\partial}{\partial h} \left[ {\rm Pr} \left(H \leq h \right) \right]  \\ 
&= \frac{\partial}{\partial h} \left[ \int_{d_{\min}^*}^{d_{\max}} \int_{0}^{\frac{\pi}{2}} I_{H} \left(\theta,d \right) f_{\theta}(\theta) f_{d}(d) \mathrm{d}\theta \mathrm{d}d \right] + \frac{\partial}{\partial h} \left[ F_{\cos(\psi)}(\cos(\Psi_c)) \mathcal{U}_{[0,\infty]}(h)  \right] \\ 
&=  \int_{d_{\min}^*}^{d_{\max}} \int_{0}^{\frac{\pi}{2}} \frac{d^{m+3}}{a(\theta)\sqrt{d^2-(h_a-h_u)^2}} f_{\cos (\Omega-\alpha)} \left( \frac{d^{m+3}h-b(\theta)}{a(\theta)\sqrt{d^2-(h_a-h_u)^2}} \right) f_{\theta}(\theta) f_{d}(d) \mathrm{d}\theta \mathrm{d}d  \\ 
&\quad + v(h) \int_{0}^{\frac{\pi}{2}} J_{H} \left(\theta,d \right) f_{\theta}(\theta) \mathrm{d}\theta +F_{\cos(\psi)}(\cos(\Psi_c)) \delta(h),
\end{align}
\end{subequations}
\noindent\makebox[\linewidth]{\rule{\textwidth}{0.4pt}}
\end{figure*} 
for $h \in \left[ h_{\min}^{*}, h_{\max}^{*} \right]$, and $0$ otherwise, such that $h_{\max}^{*} = h_{\max} \cos(\Psi_c) \in \left[h_{\min}, h_{\max} \right]$. This completes the proof.
\section{Proof of Theorem 2}
Based on the results of Theorem 1, the PDF of the LOS channel gain $H$ is given by 
\begin{equation}
\begin{split}
&f_{H} (h) \\
&= g_{H}(h) \mathcal{U} \left( h_{\min}^{*},h_{\max} \right) + F_{\cos(\psi)}(\cos(\Psi_c)) \delta(h) \\
&= \left(1-F_{\cos(\psi)}(\cos(\Psi_c))\right) f_Z(h) + F_{\cos(\psi)}(\cos(\Psi_c)) \delta(h),
\end{split}
\end{equation}
where $f_Z$ is a PDF, with support range $[h_{\min}^{*},h_{\max}]$. The PDF $f_Z$ associated to a random variable $Z$ that is expressed as $Z = XY$, where $X = \frac{c}{d^{m+2}}$ such that $c = H_0\left(h_{\rm a} - h_{\rm u} \right)^{m}$ and $Y = \cos \left(\psi \right)$ for $\psi \in \left[0,\Psi_c \right]$. Note that $X$ and $Y$ are two random variables that reflect the effects of the random spatial distribution of the LiFi user and the random orientation of the UE on the LOS channel gain, respectively. For the case of stationary users, and using the PDF transformation of random variables, the PDF of the random variable $X$ is expressed as 
\begin{equation}
\begin{split}
f_X(x) &= \frac{c^{\frac{1}{m+2}}}{(m+2)} \left(\frac{1}{x}\right)^{\frac{m+3}{m+2}} f_d \left(\left(\frac{c}{x} \right)^{\frac{1}{m+2}} \right) \\ 
&= \frac{2c^{\frac{2}{m+2}}}{R^2(m+2)} \left(\frac{1}{x}\right)^{\frac{m+4}{m+2}} \mathcal{U}_{\left[ c/{d_{\max}^{m+2}},c/{d_{\min}^{m+2}} \right]}(x).
\end{split}
\end{equation}
Obviously, $X$ and $Y$ are correlated since they are a function of the distance $d$, which is also a random variable. However, such correlation can be weak in most of the cases. In fact, for high values of $d$, the effect of random orientation on the LOS channel gain $H$ is negligible compared to the one of the distance, whereas for the case of low values of $d$, the effect of distance is negligible compared to the one of random orientation. Due to this, as an approximation, we assume that the random variables $X$ and $Y$ are uncorrelated. Based on this, using the theorem of the PDF of the product of random variables \cite{glen2004computing}, the PDF of the random variable $Z$ can be approximated as
\begin{subequations}
\begin{align}
f_Z(h) &\approx \int_{y_{\min}(h)}^{y_{\max}(h)} f_X \left( \frac{h}{y} \right) f_Y \left(y \right) \frac{\mathrm{d}y}{y}  \\ 
&= \int_{y_{\min}(h)}^{y_{\max}(h)} \frac{2c^{\frac{2}{m+2}}}{R^2(m+2)} \left(\frac{y}{h}\right)^{\frac{m+4}{m+2}}  f_Y \left(y \right) \frac{\mathrm{d}y}{y} \\
&= \left(\frac{1}{h}\right)^{\frac{m+4}{m+2}} \int_{y_{\min}(h)}^{y_{\max}(h)} \frac{2c^{\frac{2}{m+2}}}{R^2(m+2)} y^{\frac{2}{m+2}}   f_Y \left(y \right) \mathrm{d}y,
\end{align}
\end{subequations}
where $f_{Y}$ denotes the PDF of the random variable $Y$ and it is given in (19) and (22) of \cite{soltani2018modeling}. Based on this, the PDF $f_Z$ has the form 
\begin{equation}
    f_Z(h) \approx \frac{1}{h^{\nu}}  \Tilde{f} (h), 
\end{equation}
where $\nu > 0$ and $\Tilde{f}$ is a function with support range $[h_{\min}^{*},h_{\max}]$ that is expressed as
\begin{equation}
    \Tilde{f}(h) = \int_{y_{\min}(h)}^{y_{\max}(h)} \frac{2c^{\frac{2}{m+2}}}{R^2(m+2)} y^{\frac{2}{m+2}}   f_Y \left(y \right) \mathrm{d}y.
\end{equation}
Consequently, by substituting $f_Z(h)$ in (47) by its expression and defining the function $g$, for $h \in [h_{\min}^{*},h_{\max}]$, as 
\begin{equation}
    g(h) = \left(1-F_{\cos(\psi)}(\cos(\Psi_c))\right) \Tilde{f}(h),
\end{equation}
we obtain the result of Theorem 2, which completes the proof.
\section{Proof of Theorem 3}
Using the same notation adopted in Appendix B, the PDF of the random variable $X$ for the case of mobile users is expressed as 
\begin{equation}
\begin{split}
f_X(x) &= \frac{c^{\frac{1}{m+2}}}{(m+2)} \left(\frac{1}{x}\right)^{\frac{m+3}{m+2}} f_d \left(\left(\frac{c}{x} \right)^{\frac{1}{m+2}} \right) \\ 
&= \sum_{i=1}^3 \frac{a_i c^{b_i + \frac{2}{m+2}}}{(m+2)R^{b_i+1}} \left(\frac{1}{x}\right)^{b_i + \frac{m+4}{m+2}}\mathcal{U}_{\left[ c/{d_{\max}^{m+2}},c/{d_{\min}^{m+2}} \right]}(x).
\end{split}
\end{equation}
Therefore, from (49b), the PDF of the random variable $Z$ can be approximated by
\begin{equation}
\begin{split}
f_Z(h) &\approx \int_{y_{\min}(h)}^{y_{\max}(h)} f_X \left( \frac{h}{y} \right) f_Y \left(y \right) \frac{\mathrm{d}y}{y}  \\ 
&= \sum_{i=1}^3 \left(\frac{1}{h}\right)^{b_i+\frac{m+4}{m+2}} \int_{y_{\min}(h)}^{y_{\max}(h)} \left[ \frac{a_i c^{b_i + \frac{2}{m+2}}}{(m+2)R^{b_i+1}} \right. \\ 
&\qquad \qquad \qquad \qquad \qquad \qquad \left. \times y^{b_i+\frac{m+4}{m+2}} f_Y \left(y \right) \right] \mathrm{d}y,
\end{split}
\end{equation}
which has the form 
\begin{equation}
    f_Z(h) \approx \sum_{j=1}^3 \frac{1}{z^{\nu_i}}  \Tilde{f}_{j} (h), 
\end{equation}
where, for $j=1,2,3$, $\nu_j > 0$ and $\Tilde{f}_{j}$ is a function with support range $[h_{\min}^{*},h_{\max}]$ that is expressed as
\begin{equation}
    \Tilde{f}_{j}(h) = \int_{y_{\min}(h)}^{y_{\max}(h)} \left[ \frac{a_i c^{b_i + \frac{2}{m+2}}}{(m+2)R^{b_i+1}} y^{b_i+\frac{m+4}{m+2}} f_Y \left(y \right) \right] \mathrm{d}y.
\end{equation}
Consequently, by substituting $f_Z(h)$ in (47) by its expression and defining $g_{j}(h)$, for $j=1,2,3$ and for $h \in [h_{\min}^{*},h_{\max}]$, as 
\begin{equation}
    g_{j}(h) = \left(1-F_{\cos(\psi)}(\cos(\Psi_c))\right) \Tilde{f}_{j} (h),
\end{equation}
we obtain the result of Theorem 3, which completes the proof.
\section{Proof of Corollary 1}
Based on PDF of the LOS channel gain $H$ provided in Theorem 1 and the expression of the probability of error of $M$-pulse amplitude modulation \cite{Intro13,luo2014fundamental,yeh1998approximate}, the average probability of error of the considered LiFi system is expressed as 
\begin{subequations}
\begin{align}
P_e  &= \int_{h_{\min}}^{h_{\max}} f_{H} (h) P_{e,h}\left(P_{\rm opt},h\right) \mathrm{d}h  \\ 
&= \int_{h_{\min}}^{h_{\max}} f_{H} (h) \frac{2(M-1)}{M} \mathcal{Q} \left(\frac{h\sqrt{\gamma_{\rm TX}}}{M-1} \right) \mathrm{d}h \\ 
&= \int_{h_{\min}^{*}}^{h_{\max}} g_{H} (h) \frac{2(M-1)}{M} \mathcal{Q} \left(\frac{h P_{\rm opt} }{\sigma \left(M-1\right)} \right) \mathrm{d}h  \\ 
&+ F_{\cos(\psi)}(\cos(\Psi_c)) \int_{h_{\min}}^{h_{\max}} \delta(h) \mathcal{Q} \left(\frac{h P_{\rm opt} }{\sigma \left(M-1\right)} \right) \mathrm{d}h \nonumber,
\end{align}
\end{subequations}
where $P_{e,h}\left(P_{\rm opt},h\right)$ is the instantaneous probability of error for a given channel gain $h$ and $\gamma_{\rm TX} = \frac{P_{\rm elec}}{\sigma^2} = \frac{P_{\rm opt}^2}{\sigma^2}$ is the transmitted signal to noise ratio, such that $P_{\rm elec}$ is the transmitted signal to noise ratio and $\sigma^2$ is the average noise power at the receiver. Now, since the function $h \mapsto g_{H} (h) \frac{2(M-1)}{M} \mathcal{Q} \left(\frac{h P_{\rm opt} }{\sigma \left(M-1\right)} \right)$ is a smooth function within $\left[h_{\min}^{*}, h_{\max} \right]$, and using the Lebesgue's dominated convergence theorem, we get 
\begin{equation}
    \begin{split}
        &\lim_{P_{\rm opt} \rightarrow \infty} \int_{h_{\min}}^{h_{\max}} g_{H} (h) \frac{2(M-1)}{M} \mathcal{Q} \left(\frac{h P_{\rm opt} }{\sigma \left(M-1\right)} \right) \mathrm{d}h \\ 
        &\qquad =\int_{h_{\min}}^{h_{\max}} \lim_{P_{\rm opt} \rightarrow \infty} g_{H} (h) \frac{2(M-1)}{M} \mathcal{Q} \left(\frac{h P_{\rm opt} }{\sigma \left(M-1\right)} \right) \mathrm{d}h \\
        &\qquad = 0.
    \end{split}
\end{equation}
Furthermore, we have $\int_{h_{\min}}^{h_{\max}} \delta(h) \mathcal{Q} \left(\frac{h P_{\rm opt} }{\sigma \left(M-1\right)} \right) \mathrm{d}h = \mathcal{Q} \left( 0 \right) = \frac{1}{2}$ since $h_{\min} = 0$, which implies that $\lim_{P_{\rm opt} \rightarrow \infty} \int_{h_{\min}}^{h_{\max}} \delta(h) \mathcal{Q} \left(\frac{h P_{\rm opt} }{\sigma \left(M-1\right)} \right) \mathrm{d}h = \frac{1}{2}$. Therefore, we conclude that $\lim_{P_{\rm opt} \rightarrow \infty}  P_e \left(P_{\rm opt} \right) = \frac{F_{\cos(\psi)}(\cos(\Psi_c))}{2}$, which completes the proof.
\bibliographystyle{IEEEtran}
\bibliography{journalbiblio}
\end{document}